\documentclass[11pt]{article}
\usepackage{amsmath,amssymb,theorem,enumerate,graphicx,multirow,multicol,diagbox}
\usepackage{array}

\author{Marc Demange\footnote{RMIT University, School of  Science, Australia; e-mail: {\tt marc.demange@rmit.edu.au}} 
\quad
Dominique de Werra\footnote{Ecole Polytechnique F\'ed\'erale de Lausanne, EPFL, Switzerland; e-mail: {\tt dominique.dewerra@epfl.ch}}}
\vspace{3 cm}

\title{Complexity of choosability with a small palette of colors\thanks{This research was carried out while the first author was visiting EPFL. The support of this institution is gratefully acknowledged.
}}
\date{April 17, 2017}

\newtheorem{theorem}{Theorem}[section]
\newtheorem{claim}{Claim}[section]
\newtheorem{corollary}{Corollary}[section]
\newtheorem{lemma}{Lemma}

\newtheorem{proposition}{Proposition}[section]
\newtheorem{remark}{Remark}[section]

\newtheorem{pf claim}{Proof of Claim}
\newenvironment{proof}[1][Proof]{\noindent \textbf{#1. }}{\ \rule{0.5em}{0.5em}}

\begin{document}

\begin{titlepage}
\maketitle
\thispagestyle{empty}
\noindent

\begin{abstract}
A graph is $\ell$-choosable if, for any choice of lists of $\ell$ colors for each vertex, there is a list coloring, which is a coloring where each vertex receives a color from its list. We study complexity issues of choosability of graphs when the number $k$ of colors is limited. We get results which differ surprisingly from the usual case where $k$ is implicit and which extend known results for the usual case. We also exhibit some classes of graphs (defined by structural properties of their blocks) which are choosable. 
\end{abstract}


\end{titlepage}

\section{Introduction}

List coloring and choosability in graphs have been extensively studied these last years. The literature provides a huge collection of results for most of which the number of colors is implicit. Here, we intend to explore choosability and list coloring issues in which the number of available colors (denoted by $k$ in the sequel) is fixed. To our knowledge, this idea was introduced in~\cite{cogis} and~\cite{ganjali} and further studied in particular in~\cite{kral} and more recently in~\cite{kang} where the set of available colors is called the spectrum and in~\cite{bonamy} where it is called the palette.  It turns out that limiting the size of the palette leads to some original and unexpected results. We shall derive below some results in this direction, in particular related to complexity.

The motivation for such research originates in some types of scheduling problems. A time interval of length $k$ is divided into $k$ unit time periods. We have a collection $V$ of tasks with unit processing times to schedule in this time interval. Each task $v$ to be scheduled has a collection $L(v)$ of possible time periods for its processing and should be assigned to exactly one period in the list $L(v)$, while taking into account a set $E$ of pairwise incompatibilities between tasks: some pairs of tasks cannot be processed during the same period because they may use a common resource (like a machine). It is usual in  such applications to assume that the total number of possible periods is fixed to some value $k$. 

More formally, we are given a graph $G=(V,E)$ and a palette ${\cal K}=\{1,2, \ldots, k\}$ of colors. Each vertex $v$ is associated with a list $L(v)\subseteq {\cal K}$ of possible colors. We denote such data by $(G,L)$ where $L$ will be called a list assignment, or list system. One has to find whether there is a vertex $k$-coloring $G$ which gives to each vertex $v$ a color $c(v)\in L(v)$ in such a way that $c(u)\neq c(v)$ for every edge $uv\in E$. Such a coloring is a {\em list $k$-coloring}.  

An interesting situation occurs when the size of lists are all equal to a fixed constant $\ell$; this happens in particular when one allows a fixed number of time preferences for each task to give more flexibility in the construction of the schedule. In this work we are interested in finding  whether or not such a coloring exists for any list system $L$ satisfying $|L(v)|=\ell$ for each $v\in V$. This property, called {\em choosability}, is particularly relevant since it  indicates that whatever the lists will be, it will be possible to find a feasible schedule. Then this schedule can be built using other criteria than the strict limitation of available colors. This can also be a meaningful information when deciding the number $\ell$ of alternatives to be chosen for each task.

After recalling some basic results on choosability and list coloring in Section~\ref{sec:notations}, we exhibit in Section~\ref{sec:polynomial} some classes of choosable graphs defined by blocks for which list colorings can be constructed in polynomial time. We discuss new complexity results in Section~\ref{sec:hardness} for the problem of deciding whether a graph is choosable, these results hold for small color palettes and restricted classes of graphs; 
We discuss some additional remarks about complexity in Section~\ref{sec:remarks} and finally, Section~\ref{sec:conclusion} presents a table summarizing the main results, concluding remarks as well as some open questions.

\section{Notations}\label{sec:notations} 

All graphs considered in this work are simple and finite.
Given a graph $G$, a vertex $k$-coloring assigns to every vertex $v$ a color $c(v)\in\{1, \ldots, k\}$ with $c(u)\neq c(v)$ for every edge $uv\in E$.
Given a graph $G$ and a list assignment $L$ with a palette of colors ${\cal K}$ ($k=|{\cal K}|)$, a {\em list $k$-coloring} is a $k$-coloring such that each vertex gets a color in its list. $L$ is {\em feasible} if there is a list $k$-coloring for $L$ while it is {\em infeasible} in the opposite case.

Given a graph $G=(V,E)$ and a function $f: V\longrightarrow {\mathbb N}$, an $f$-list assignment $L$ is a list assignment satisfying $|L(v)|=f(v), \forall v\in V$. 
For an edge $uv$, $G\setminus \{uv\}$ denotes the graph obtained from $G$ by removing $uv$ from the edge set. We will denote the neighborhood of a vertex $u$ by $\Gamma(u)$. For a set of vertices $X$, $G[X]$ denotes the subgraph of $G$ induced by $X$. A graph $G'=(V,E')$ with $E'\subset E$ is called {\em partial graph} of $G$ and a {\em partial induced subgraph} is a partial graph of an induced subgraph.  A set of vertices is {\em stable} if it consists of pairwise non-adjacent vertices. A {clique} $K_m$ is a complete graph on $m$ vertices; $K_{p,q}$ denotes a complete bipartite graph with parts of size $p$ and $q$ and $K_{p,q,r}$ denotes a complete tripartite graph with parts of size $p$, $q$ and $r$. $C_q$  denotes an elementary cycle of size $q$ (connected graph with $q$ vertices, all  of degree~2). An elementary path that is not a cycle has two end vertices of degree~1 and all other vertices of degree~2, called {\em internal}; the length is the number of edges. We denote by $G(p,q)$ the grid-graph (or shortly grid) of size $p\times q$ with vertex set $V = ((i,j) i = 1 \dots p; j = 1, \dots q)$ and edge set
$$
\begin{array}{rcl}
E = && \{ \left[ (i,j),(i,{j+1}) \right] i = 1, \dots,p,\  j = 1, \dots, q- 1 \} \\
& \bigcup & \{ \left[ (i,j),({i+1},j) \right],  i = 1, \dots,q-1,\  j = 1, \dots, q \}    \\
\end{array}
$$
Of course, $G(p,q)$ and $G(q,p)$ are isomorphic. $G(2,3)$ (or $G(3,2)$) is called a {\em chocolate}. By {\em subgrid} we will mean an (induced) subgraph of a grid.
All graph-theoretical terms not defined here can be found in~\cite{cb73}. For complexity concepts we refer the reader to~\cite{gj}.

\vspace{0.5 cm}

Given  a graph $G=(V,E)$ and a function $f: V\longrightarrow {\mathbb N}$, $G$ is called {\em $[f, k]$-choosable} if it has a list $k$-coloring for every list system $L$ satisfying $\forall v\in V, |L(v)|=f(v)$. 
If $\forall v\in V, f(v)=\ell$, then $G$ is simply called  {\em $[\ell, k]$-choosable}. 

We assume that colors in each list $L(v)$ are distinct and consequently that the size $k$ of the palette satisfies $k\geq f$, i.e., $k\geq \max (f(v), v\in V)$.

A graph is {\em $f$-choosable} (resp. {\em $\ell$-choosable}) if it is $[f,k]$-choosable for any $k\geq f$ (resp. if it is $[\ell,k]$-choosable for any $k\geq \ell$). 
We must also have $k\geq \chi(G)$ for such a coloring to exist, where $\chi(G)$ is the chromatic number of the graph $G$.

Note that a graph is $[f,k]$-choosable (resp. $f$-choosable, $\ell$-choosable) if and only if it is the case for each connected component.

\begin{remark}
A similar notion of choosability with $k$ colors, denoted by $(\ell,k)$-choosability, has been used for instance in~\cite{ktchoos} and~\cite{ganjali}. The only difference is that in these works $k$ denotes the number of colors in the union of all lists. Notice that with such a definition, choosability with $k$ colors does not imply choosability with $k-1$ colors. Consider for instance a two-edge path on vertices $x,y,z$; it would be 1-choosable if $|L(x)\cup L(y)\cup L(z)|=3$ while it is trivially not 1-choosable with only two colors.

The same notation was used in~\cite{kral} for a concept identical with our definition of $[\ell,k]$-choosability. We chose another notation to avoid any confusion.
\end{remark}

Our definition does not require that a list assignment uses all colors of the palette, which ensures that $[f,k]$-choosability 
implies $[f,k']$-choosability for any $k'$ such that $\max (f(v), v\in V)\leq k'\leq k$. Moreover, $[f,k]$-choosability 
implies also $[f',k]$-choosability 
for any $f'$ such that  $k\geq f'\geq f$. Note that $[k,k]$-choosability 
is exactly the $k$-colorability and that $[f,k]$-choosability 
is an hereditary property.

 Note also that the concept of $[\ell,k]$-choosability, based on usual graph colorings, has been generalized in~\cite{kang} to the case of $t$-improper colorings where each color class induces a subgraph of maximum degree at most~$t$.

In~\cite{gutner}, a notion of $\ell$-choice criticality was introduced. Here we give another definition where the size of the palette  appears explicitly. In Section~\ref{sec:hardness} it will allow us to derive hardness results for smaller and sometimes optimal palette sizes.

Given a graph $G=(V,E)$, a function $f: V\longrightarrow {\mathbb N}$ and $V'\subset V$, $G$ is  {\em $([f,k], V')$-critical} if there is an infeasible $f$-list assignment $L$ such that $\forall v\in V, L(v)\subset \{1, \ldots, k\}$ and $|\cup_{v\in V'}L(v)|\leq k-1$  but if we  replace $f(v)$ by $f(v)+1$ for an arbitrary vertex $v\in V'$, then $G$ becomes $[f,k]$-choosable. Note that if a graph is $\ell$-choice critical, as defined in~\cite{gutner}, then there is a vertex $v_0$ and a palette size~$k$ such that  $G$ is $([f,k],\{v_0\})$ critical  with$f(v)=\ell, \forall v\neq v_0$ and $f(v_0)=\ell-1$. 

$[\ell, k]$-LISTCOL will denote the following problem: given an instance $(G,L)$, where $G=(V,E)$ is a graph and $L$ a list assignment such that $\forall v\in V, 
L(v)\subseteq \{1, \ldots, k\}$ and $|L(v)|=\ell$, decide whether there is a  $k$-coloring satisfying all list requirements. 
$\ell$-LISTCOL is defined in a similar way when no limitation is imposed on the total number of available colors.  For a fixed set $\Lambda$ of positive integers, $[\Lambda, k]$-LISTCOL  is the similar problem with the condition that all  list sizes are chosen in $\Lambda$. So, $[\ell, k]$-LISTCOL is the simplified notation for $[\{\ell\}, k]$-LISTCOL.

$[\ell, k]$-CH (resp. $\ell$-CH) denotes the problem of deciding whether a given graph $G$ is $[\ell, k]$-choosable (resp. $\ell$-choosable). For a fixed set $\Lambda$ of positive integers, {\em $[\Lambda, k]$-CH} is the problem defined as follows: an instance is a graph $G$ and a function $f: V(G)\rightarrow \Lambda$ and the problem is to decide whether $G$ is $[f,k]$-choosable. 

For a bipartite graph $G=(B\cup W, E)$, we often consider the special case where all lists in each part have the same size, which corresponds to the function $f$ defined by: 
\begin{eqnarray}\label{eq:pq}
f(v)=
\left\{
\begin{array}{lll}
p & \mbox{for} &  v \in B \\
q & \mbox{for} &  v \in W
\end{array}
\right.
\end{eqnarray}

For bipartite graphs  
{$[f,k]$-choosability}, $f$-choosability and\\ $([f,k], V')$-criticality will be respectively called  $[(p,q),k]$-choosability, $(p,q)$-choosability and $([(p,q),k], V')$-criticality if $f$ is defined by (\ref{eq:pq}). 
Similarly we define the problems $[(p,q),k]$-LISTCOL, $(p,q)$-LISTCOL, $[(p,q),k]$-CH and $(p,q)$-CH whose instances are bipartite graphs $(B\cup W,E)$ with lists of size $p$ in $B$ and $q$ in $W$. We shall always assume that $k\geq \max(p,q)$.

\section{Classes of choosable graphs defined by blocks}\label{sec:polynomial} 

In this section we will present some classes of graphs defined by structural properties of their blocks. We will in particular examine how choosability properties of their blocks can be extended to the entire graph and we will show that the related list coloring problem can be solved in polynomial time. 
We recall that a {\em cut vertex} in a graph is a vertex whose removal disconnects the graph. A  {\em block} is a maximal connected subgraph without cut vertex. 

A graph will be called {\em ideally} $[\ell,k]$-choosable if for any 
$\ell$-list assignment with colors in $\{1, \ldots, k\}$ and any vertex $v_0$, if one assigns to $v_0$ any color   $c\in L(v_0)$ it is always possible to extend the list coloring to the entire  $G$. 

\begin{proposition}\label{prop: block}
Assume in each connected component of a graph $G$ all blocks but one are ideally $[\ell,k]$-choosable (resp. $\ell$-choosable) and the last one is $[\ell,k]$-choosable (resp. $\ell$-choosable), then $G$ itself is $[\ell,k]$-choosable (resp. $\ell$-choosable).
\end{proposition}

\begin{proof}
We assume $G$ is connected, otherwise we apply the following to each connected component.

Let ${\cal B}=\{B_1, \ldots, B_b\}$ be the set of blocks of $G$, where $B_2, \ldots, B_b$ are ideally $[\ell,k]$-choosable (resp. $\ell$-choosable) and $B_1$ is $[\ell,k]$-choosable (resp. $\ell$-choosable). Let ${\cal C}=\{c_1, \ldots, c_q\}$ be the set of cut vertices of $G$. We construct the bipartite graph $\hat G=({\cal B},{\cal C},E)$, where $B_jc_i\in E$ if $c_i\in B_j$.

		It is known that $\hat G$ is a tree, called the {\em block-tree} since $G$ is connected~\cite{diestel}. Hence ordering ${\cal B}$ using the breadth-first order from the root $B_1$,  
$B_2, \ldots, B_b$ are ordered in such a way that, for any $j, 2\leq j\leq b$, $B_j$ has exactly one vertex, say $v_j$, with $v_j\in B_1\cup \ldots \cup B_{j-1}$:  $v_j$ is the only predecessor of $B_i$ when directing the edges from $B_1$ to the leaves. 

Consider any list assignment. A suitable coloring can always be constructed for $B_1$ from the assumptions. Now, having found a suitable coloring for the vertices of $B_1\cup \ldots \cup B_{j-1}$ $(j\geq 2)$ we impose the color $c(v_j)$ given to vertex $v_j\in B_1\cup \ldots \cup B_{j-1}$ and since $B_j$ is ideally $[\ell,k]$-choosable (resp. $\ell$-choosable) we can extend the list coloring to the vertices of $B_j$. By repeating this until the vertices of $B_b$ are colored, we get a suitable coloring of $G$. 
\end{proof}

\begin{remark}\label{rem:blocks}
The blocks of a connected graph together with the associated block-tree can be found in $O(|V|+|E|)$~\cite{gondran}. The breadth-first order   of this tree can be obtained in $O(|V|)$ since there are at most $|V|$ blocks. For a non connected graph, the connected components as well as the related block-trees and their breadth-first order can be computed in $O(|V|+|E|)$. We then obtain the following consequence of the proof of Proposition~\ref{prop: block}.

Let ${\cal C}$ be a class of connected
 graphs which are ideally $[\ell,k]$-choosable (resp. $\ell$-choosable) and for which $[\ell,k]$-LISTCOL (resp. $\ell$--LISTCOL) is polynomial of complexity $O(T)$. Then $[\ell,k]$-LISTCOL (resp. $\ell$--LISTCOL) is polynomial of complexity $O(\max(T, |V|+|E|))$ for the class of graphs whose blocks are in ${\cal C}$. 
\end{remark}

Line-perfect graphs  are the graphs $G$ for which the line graph $L(G)$ is perfect~\cite{trotter-line-perfect}. They are characterized by the fact that their blocks are isomorphic to $K_4$, $K_{1,1,p}, p\geq 1$ or any 2-connected bipartite graph~\cite{maffray}.

More generally, {\em quasi line-perfect}  graphs are graphs whose blocks are isomorphic to $K_4$, a 2-connected bipartite graph, $K_{1,1,p}$ or an odd cycle $C_{2p+1}, (p\geq 1)$. 
Note in particular that line-perfect graphs and cacti (blocks are cycles or edges) are quasi line-perfect.

In the following we derive some results related to the different types of blocks involved in quasi line-perfect graphs.

\begin{proposition}\label{claim: A}
Bipartite graphs are ideally $[3,4]$-choosable and $[3,4]$-LISTCOL can be solved in $O(|V|)$ in this class. 
\end{proposition}
\begin{proof}
Let $G=(B\cup W, E)$ be a bipartite graph, assume $v_0\in B$ and choose $c\in L(v_0)$. Color with $c$ all vertices  $v\in B$ such that $c\in L(v)$; since lists are of length~3 and the palette includes only four colors, all other vertices in $B$ have the same list and  all can be colored with the same color $c'$. Since $|L(u)|=3$ for each $u\in W$, we can find in $L(u)$ a color $c(u)\neq c,c'$, which allows to extend the list coloring to~$W$. The related complexity is $O(|V|)$.
\end{proof}

\begin{proposition}\label{claim: C}
Let $G$ be a graph with an order  $v_1, \ldots, v_n$ of its vertices and the related acyclic orientation of its edges ($v_iv_j$ oriented from $v_i$ to $v_j$ if $i<j$). Let $d^-_G(v_i)$ be the number of edges oriented towards $v_i$, $i=1, \ldots, n$. Consider  a list assignment $L$ with colors in $\{1, \ldots, k\}$. If  $\forall i=1, \ldots, n, |L(v_i)|\geq d^-_G(v_i)+1$, then for any choice of color $c\in L(v_1)$ for $v_1$, there is a list $k$-coloring with $v_1$ colored $c$. It can be constructed in $O(|V|+|E|)$.
\end{proposition}
\begin{proof}
$v_1$ is colored $c$; consider vertices in the order $v_2, \ldots, v_n$ and color $v_i$ with the first available color in $L(v_i)$ (not used for coloring its predecessors in the order). This is always possible by the assumptions. The complexity is $O(|V|+|E|)$. 
\end{proof}

\begin{corollary}~\label{cor:acyclic}
If  $G=K_{1,1,p}$ ($p\geq 1$), then $G$ is ideally $3$-choosable and $3$-LISTCOL can be solved in $O(|V|)$.
\end{corollary}
\begin{proof}
$K_{1,1,p}$ consists of a clique $K_2=\{a,b\}$ completely linked to a stable set $S_p$ of size~$p$.
Choose an arbitrary vertex $v_1$ and consider any order $v_2, v_3, \ldots, v_n$ of other vertices such  that $\{a,b\}\subset \{v_1, v_2, v_3\}$. Orienting each edge $v_iv_j, i<j$ from $i$ to $j$ we have $\forall i=1, \ldots, n$, $d^-_G(v_i)\leq 2$. 
Using Proposition~\ref{claim: C}, we conclude that $K_{1,1,p}$  is ideally $3$-choosable. The related complexity is $O(|V|)$ since $|E|=2|V|-3$.
\end{proof}

\begin{corollary}\label{claim: B}
A graph  $G$ of maximum degree $\Delta$ is ideally $(\Delta+1)$-choosable and $(\Delta+1)$-LISTCOL can be solved in $O(|V|+|E|)$. In particular a clique $K_n$ is ideally $n$-choosable and a cycle $C_n$ is ideally 3-choosable.
\end{corollary}
\begin{proof}
Given $G$, a $(\Delta+1)$-list assignment and an arbitrary vertex $v_1$,  by taking any order $v_1, \ldots, v_n$ and the corresponding  orientation, the result follows from Proposition~\ref{claim: C}.
\end{proof}

Note that an even cycle $C_{2p}$ is 2-choosable~\cite{rubin} but not ideally 2-choosable. 

Consider any graph $G$ and compute the list of its blocks and the block-tree in $O(|V|+|E|)$~\cite{gondran}. Denoting by $D$ the maximum degree of the blocks, we  deduce from Corollary~\ref{claim: B} that blocks are $(D+1)$-ideally choosable and then we can use Proposition~\ref{prop: block} and Remark~\ref{rem:blocks} to derive the following:

\begin{proposition}
Consider a graph $G$ and let $D$ be the maximum degree of its blocks, then $G$ is $(D+1)$-choosable and $(D+1)$-LISTCOL can be solved in $O(|V|+|E|)$.
\end{proposition}

We derive the following for quasi-line perfect graphs:

\begin{proposition}\mbox{}\\
(i) Quasi-line perfect graphs are 4-choosable and 4-LISTCOL is solvable in $O(|V|+|E|)$ in this class.\\
(ii) 3-colorable quasi-line perfect graphs are $[3,4]$-choosable and $[3,4]$-LISTCOL is solvable in $O(|V|+|E|)$ in this class.
\end{proposition}

\begin{proof}
This follows immediately from Propositions~\ref{prop: block} and~\ref{claim: A}, Remark~\ref{rem:blocks}  and from Corollary~\ref{claim: B} noticing that 3-colorable quasi-line perfect graphs are exactly graphs whose blocks are isomorphic to a 2-connected bipartite graph, $K_{1,1,p}$ or an odd cycle $C_{2p+1}, (p\geq 2)$.
\end{proof}

We also mention that {\em{ block-cactus}} graphs have been defined as graphs whose blocks are either cliques or cycles of length at least 
four~\cite{block-cactus}. We derive the following:

\begin{proposition}\mbox{}\\
Let $G$ be a block-cactus graph, let $p$ be the maximum size of a clique in $G$.\\
(i) If no block is an odd cycle $C_{2q+1}$ and each connected component has at most one even cycle $C_{2q}$ as a block, then $G$ is $p$-choosable.\\
(ii) In all other cases, $G$ is $\max(p,3)$-choosable.
\end{proposition}
\begin{proof}
Based on Corollary~\ref{claim: B}, a clique $K_p$ is ideally $p$-choosable and a cycle $C_q$	is ideally 3-choosable. If $q$ is even, it is also 2-choosable~\cite{rubin}.

We then apply Prop~\ref{prop: block}. Note that in case (i), if there is at least one connected component with  one even cycle $C_{2q}$ as a block, then $p\geq 2$.
\end{proof}

\section{Hardness of $[\ell,k]$-choosability}\label{sec:hardness} 

To our knowledge hardness of choosability has been mainly studied in~\cite{rubin,gutner} and~\cite{gutner-tarsi}. These results are derived without any consideration of the size of the palette. Moreover they essentially deal with   bipartite and/or planar graphs. 

In this section we derive some hardness results for $[\ell,k]$-choosability for $\ell\leq 3$, small palette size~$k$ and restricted classes of graphs. It improves known hardness results.
A popular research direction in choosability is to find classes of choosable graphs. For these classes the related choosability decision problem is trivial. In particular, one objective of this section is to find  "minimum"  extensions of these classes where deciding whether graphs are choosable becomes hard.

Note that, in general, $[\ell,k]$-CH does not a priori belong to NP $\cup$ co-NP but is in the class $\Pi_2^p$ (see~\cite{papadimitriou} for definitions related to complexity theory). Indeed, to check whether a fixed graph is choosable (with given list sizes), it is enough to check for any list system satisfying the size constraints  whether the related LISTCOL instance is satisfiable, a problem in NP.

In very special cases we may have  choosability problems which are in NP (see Proposition~\ref{prop: 2,3-choos 3 col}) or in co-NP (see Proposition~\ref{prop: co-np}).\\

\subsection{2-choosability with bounded palette}\label{subsec:2 choos}

2-choosability has been completely characterized in~\cite{rubin}. The {\em core} of a graph is the subgraph obtained by repeatedly  removing a vertex of degree~1 together with its incident edge until the graph contains only isolated vertices and vertices of degree at least~2. The following result is the basis of the characterization in~\cite{rubin}:
\begin{claim}\label{claim:core}
	For any graph $G=(V,E)$ and any function $f: V\longrightarrow {\mathbb N}\setminus \{0,1\}$ assigning to every vertex a value at least~2, $G$ is  $[f,k]$-choosable if and only if its core is $[f,k]$-choosable.
\end{claim}
\begin{proof}
	Any list coloring of the core can be greedily completed to the whole graph by taking vertices in the reverse order of their removal when computing the core.
\end{proof}

  For positive integers $a,b,c,d$ we denote by $\theta_{a,b,c}$ (resp. $\theta_{a,b,c,d}$)  the graph consisting of three (resp. four) internal vertex disjoint elementary paths of  length (number of edges) $a,b,c$ (resp.  $a,b,c,d$) between two fixed vertices. For instance, $\theta_{2,2,2}=K_{2,3}$, $\theta_{2,2,2,2}=K_{2,4}$ and Figure~\ref{fig:theta-2224} represents $\theta_{2,2,2,4}$.  For positive integers $p,q$ and non negative integer $r$ we denote by $\Gamma_{p,q,r}$  
  the graph consisting of two disjoint elementary cycles $C_p$ and $C_q$ linked by an elementary path of length $r$ with all internal vertices outside the two cycles. $\Gamma_{p,q,0}$ consists of two elementary cycles $C_p$ and $C_q$ sharing a single vertex.  The following theorem is proved in~\cite{rubin}:

\begin{theorem}\label{prop: 2-4-choos}~\cite{rubin}
A connected graph is $2$-choosable if and only if its core belongs to $T=\{K_1, C_{2m+2}, \theta_{2,2,2m}, m\geq 1\}$.
\end{theorem}

The proof  is based on the two following lemmas which we will use as well for proving Theorem~\ref{prop: 2-3-choos}:

\begin{lemma}\label{lem:rubin1}~\cite{rubin}
	Let $H$ be the core of a connected graph $G$ such that $H\notin T$. Then $H$ contains a partial induced subgraph belonging to one of the following types:\\
	(1) an odd cycle $C_{2p+1}$, $p\geq 1$;\\
	(2) $\Gamma_{2p,2q,r}$, $p,q\geq 2, r\geq 0$;\\
	(3) $\theta_{a,b,c}$, $a,b \neq 2$, $c\geq 1$;\\
	(4) $\theta_{2,2,2,2m}, m\geq 1$.
\end{lemma}

\begin{lemma}\label{lem:rubin2}
	Let $G$ be a bipartite graph and $G'$ be obtained from $G$ by removing a vertex $v$, merging in a single vertex $v'$ all neighbors of $v$ and merging each family of parallel edges into a single edge. We have for any $k\geq 2$: if $G$ is $[2,k]$-choosable, then $G'$ is $[2,k]$-choosable.
\end{lemma}
\begin{proof}

The lemma is proved in~\cite{rubin} for $k\geq 4$ but the same argument is valid for $k=2,3$.  Suppose $G'$ is not $[2,k]$-choosable, $k\geq 3$ and show that $G$ is not $[2,k]$-choosable. Consider an infeasible 2-list assignment with $k$ colors 	for $G'$ and consider the following 2-list assignment with $k$ colors 	for $G$: $v$ and all its neighbors are assigned the list $L(v')$ and all other vertices have the same list as in $G'$. Then, in a list coloring of $G$ all neighbors of $v$ have necessarily the same color and thus, a list coloring for $G$ would induce a list coloring for $G'$.
\end{proof}

Using Lemmas~\ref{lem:rubin1} and~\ref{lem:rubin2}, the end of the proof of Theorem~\ref{prop: 2-4-choos} reduces the list of graphs of the four types in Lemma~\ref{lem:rubin1} to four different graphs and gives obstructions (infeasible 2-list assignments) for these graphs to show that they are not 2-choosable. 
 As noticed in~\cite{kral}, since all these obstructions involve at most four colors, it turns out that:
 \begin{proposition}\label{prop: 2-4k-choos}~\cite{kral}
A graph $G$ is $2$-choosable if and only if $\exists k\geq 4$ such that $G$ is $[2,k]$-choosable.
\end{proposition}
 
 \begin{remark}\label{rem:kk+1}
 This situation \-- where there is a threshold value $k_0$ such that a graph $G$ is $\ell$-choosable if and only if it is $[\ell,k]$-choosable for a fixed $k\geq k_0$ \-- is specific to 2-choosability. Indeed, in~\cite{kral} it is shown that, for any $\ell \geq 3$ and arbitrarily large $k$, there are graphs which are $[\ell,k]$-choosable but not $[\ell,k+1]$-choosable. Moreover, in~\cite{old version} we have shown that this is even the case for bipartite graphs. \end{remark}

$[2, 2]$-choosability corresponds to 2-colorability and to settle the case $\ell = 2$ we need to characterize $[2,3]$-choosable graphs. 

\begin{remark}\label{rem:K2m}
	$K_{2,m}$, $m\geq 2$ is $[2,3]$-choosable but 2-choosable only if $m\leq 3$. 
\end{remark}

Denote indeed by $B$ and $W$ the two parts of $K_{2,m}$ with $|B|=2$ and $|W|=m$ and take any 2-list assignment with palette $\{1,2,3\}$. Lists in $B$ have at least one common color $c$ and coloring $B$ with this color allows to greedily color $W$. Note however that, if the palette is $\{1,2,3,4\}$ and lists in $B$ are $\{1,2\}$ and $\{3,4\}$, then $B$ can be list colored with four different sets of two colors and if $|W|\geq 4$ it suffices to have four vertices in $W$ with respective lists $\{1,3\}$, $\{1,4\}$, $\{2,3\}$ and $\{2,4\}$ to make the graph not list colorable, which is as well a direct consequence of Theorem~\ref{prop: 2-4-choos}.

This example constitutes the key difference between $[2,3]$-choosability and 2-choosability as illustrated by the following result: 

\begin{theorem}\label{prop: 2-3-choos}
A connected graph is $[2,3]$-choosable if and only if its core belongs to $\widetilde T=\{K_1, C_{2m+2}, \theta_{2,2,2m}, K_{2,m+3}, m\geq 1\}$
\end{theorem}

\begin{proof}
The proof is an adaptation of the proof of Theorem~\ref{prop: 2-4-choos} in~\cite{rubin}. Consider a connected graph $G$ and denote by $H=(V_H,E_H)$ its core which is also connected. Using Theorem~\ref{prop: 2-4-choos}, if $H$ is in $T$, then 
$G$ is 2-choosable and thus $[2,3]$-choosable. Moreover, as mentioned in Remark~\ref{rem:K2m}, $K_{2,m}, m\geq 2$ is $[2,3]$-choosable and consequently, using Claim~\ref{claim:core}, if $H$ is in $\widetilde T$, then $G$ is $[2,3]$-choosable. 

Suppose now that $H$ is not in $\widetilde T$ and show that it is not  $[2,3]$-choosable. In particular $H$ is not in $T\subset \widetilde T$ and using Lemma~\ref{lem:rubin1}, $H$ contains, as partial induced subgraph, a graph of one of the types (1), (2), (3) or (4).
We then claim the following:

\begin{claim}\label{claim:theta2222}
If $H$ contains a $K_{2,4}=\theta_{2,2,2,2}$ (type (4) with $m=1$), is bipartite (no type (1)) and does not contain  a graph $\Gamma_{2p,2q,r}$, $p,q\geq 2$ (no type (2)), then it contains a graph $\theta_{2,2,2,2m}, m\geq 2$.
\end{claim}
 \begin{proof} (of the claim)
 Suppose  $H$ satisfies the conditions of the claim and consider, in $H$, a $\theta_{2,2,2,2}=(A\cup B,E_\theta)$ whose two parts are $A$, of size~2 and $B$, of size~4. Since $H$ is bipartite, there is no other edge in $E_H$ between two different vertices of $A\cup B$ and since $H\notin \widetilde T$, it  is not itself  a graph $K_{2,m}, m\geq 1$ and consequently there is a vertex  $v\in V_H\setminus (A\cup B)$. Since $H$ is connected of minimum degree at least~2, there is a path in $H$ starting from $v_0\in A\cup B$, passing through $v$ and that either contains a cycle or is elementary and arrives in $A\cup B$ to a vertex $v_1\neq v_0$. Since $H$ does not contain a $\Gamma_{2p,2q,r}$, $p,q\geq 2$, this is an elementary path from $v_0$ to $v_1$ with all internal vertices in $V_H\setminus (A\cup B)$. Necessarily $A=\{v_0,v_1\}$ since in all other cases  $H$ would contain a $\Gamma_{2p,2q,r}$, $p,q\geq 2$. Denote by ${\cal P}$ the set of all such elementary paths between $v_0$ and $v_1$. Since $H$ is not itself  a graph $K_{2,m}, m\geq 1$, at least one path in ${\cal P}$ is of length at least~3, which means that $H$ contains a graph $\theta_{2,2,2,p}, p\geq 3$. Since $H$ is bipartite, $p=2m, m\geq 2$, which concludes the proof of the claim. 
\end{proof}

Using Claim~\ref{claim:theta2222} and the fact that $H$ contains, as partial induced subgraph, a graph of one of the types (1), (2), (3) or (4) in Lemma~\ref{lem:rubin1}, it contains a graph of one of the types (1), (2), (3) or (4'), where (4') is $\theta_{2,2,2,2m}, m\geq 2$. We now use Lemma~\ref{lem:rubin2} to reduce the cases to be considered. If $H$ contains an odd cycle, then it is not bipartite and thus not $[2,3]$-choosable. Using the reduction in Lemma~\ref{lem:rubin2}, to check  $[2,3]$-choosability of bipartite graphs of types (2), (3) or (4') it suffices to check the  $[2,3]$-choosability of the following graphs:

\begin{quote}
	(a) $\Gamma_{4,4,1}$ or $\Gamma_{4,4,0}$; 
(b) $\theta_{3,3,1}$;
and (c) $\theta_{2,2,2,4}$.
\end{quote}

Indeed, case (2) leads to $\Gamma_{4,4,1}$ or $\Gamma_{4,4,0}$ depending on the parity of the path between the two cycles; case (3) with $a,b,c$ all odd leads to $\theta_{3,3,1}$; case (3) with $a,b,c$ all even leads to $\Gamma_{4,4,0}$ and case (4') leads to  
$\theta_{2,2,2,4}$. 

\begin{figure}[h]
\begin{center}
\includegraphics[scale=0.35]{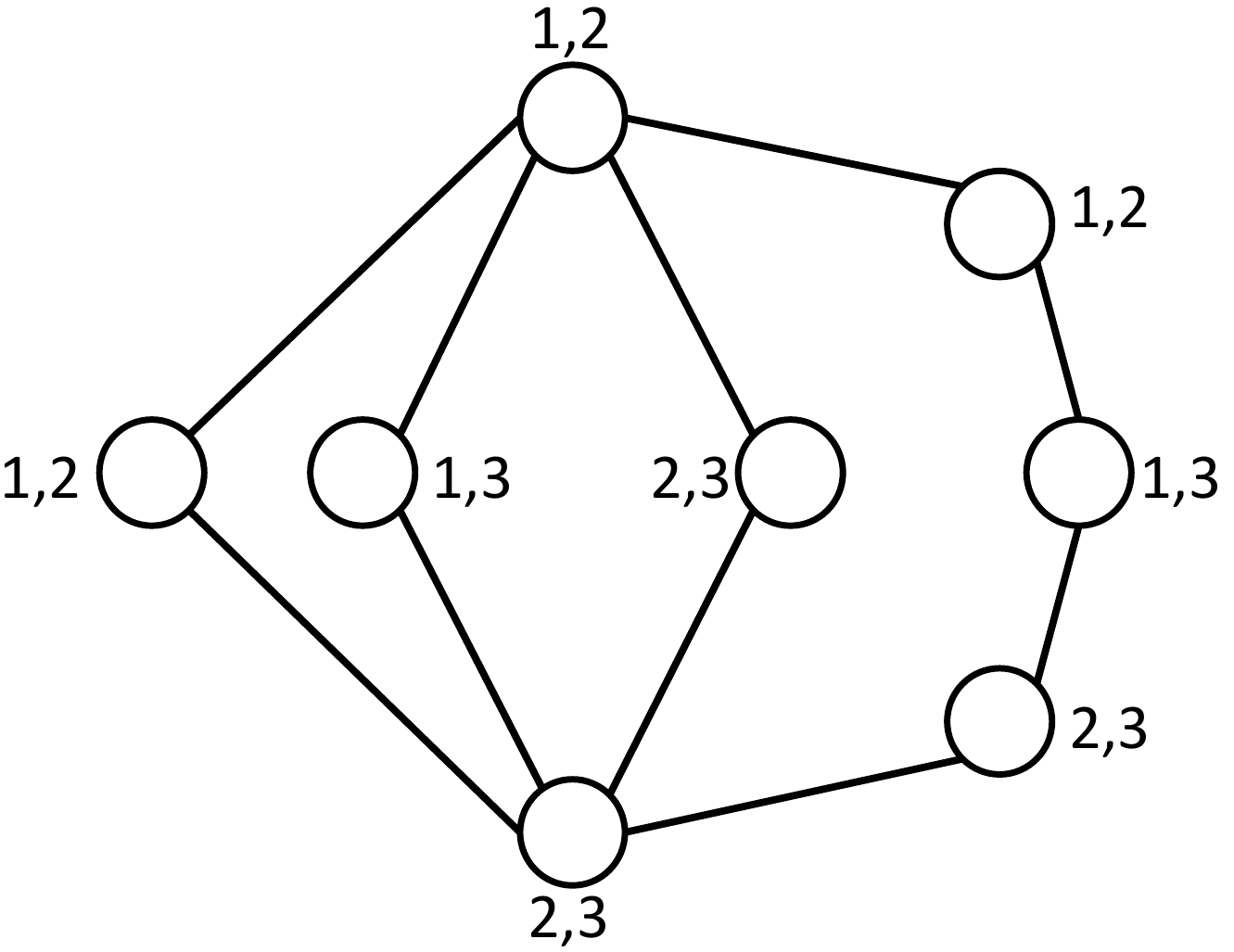}
\caption{$\theta_{2,2,2,4}$ and an infeasible 2-list assignment with three colors.}
\label{fig:theta-2224}
\end{center}
\end{figure}

It is shown in~\cite{rubin} that $\Gamma_{4,4,1},\Gamma_{4,4,0}$ and  $\theta_{3,3,1}$ are not $[2,3]$-choosable. The 2-list assignment for $\theta_{2,2,2,4}$ given in Figure~\ref{fig:theta-2224}
shows that $\theta_{2,2,2,4}$ is not $[2,3]$-choosable. This concludes the proof.
\end{proof}

We immediately deduce the following:

\begin{corollary}
	$[2,3]$-CH is polynomial.
\end{corollary}

\subsection{Extending hardness results for $\{2,3\}$-choosability}\label{subsec: 2-3-choos}

If list sizes are either 2 or 3,  
note that $[\{2,3\},3]$-CH is hard in planar graphs and in triangle-free graphs since 3-COL, equivalent to $[3,3]$-CH, is already hard in these classes~(\cite{gj,maffray-triangle-free}). 

It has been shown in~\cite{grids} that bipartite graphs are $[(2,3),3]$-choosable. The following proposition points out hard cases when list sizes are only required to be in $\{2,3\}$: 

\begin{proposition}\label{prop: 2,3-choos 3 col}\mbox{}\\
(i) $[\{2,3\},3]$-CH is NP-complete in bipartite graphs, even if there are only 6 vertices with 2-lists and they induce a $C_6$.\\
(ii) $[\{2,3\},3]$-CH is polynomial in bipartite graphs if there are 6 vertices with 2-lists and they do not induce a $C_6$.\\
(iii) Every bipartite graph is $[\{2,3\},3]$-choosable if at most 5 vertices have 2-lists.
\end{proposition}

\begin{proof}
{\bf (i)}.
Membership to NP is clear since under the assumptions  there are only $3^6$ different list systems; for each one of them one needs to  guess a coloring and check whether it is a feasible list coloring. If the answer is Yes, then the graph is choosable, else it is not. The reduction is made from 3-color Pre-Ext in bipartite graphs defined as follows: an instance is a bipartite graph $G=(B\cup W,E)$ with three specified vertices $v_1, v_2, v_3$ in $B$. The question is whether there is a 3-coloring of $G$ such that $v_1, v_2$ and $v_3$ get different colors. This problem was shown NP-complete in~\cite{bojawoe94}. 

Consider an instance of 3-color Pre-Ext consisting of a bipartite graph $G=(B\cup W,E)$ and three vertices $\{v_1, v_2, v_3\}\subset B$. We add three vertices $u_1, u_2, u_3$ in $W$ and edges $v_1u_1$, $u_1v_2$, $v_2u_2$, $u_2v_3$, $v_3u_3$ and $u_3v_1$. We denote by $C$ the cycle $v_1,u_1,v_2,u_2,v_3,u_3$. The resulting bipartite graph $G'$ is seen as an instance of $[\{2,3\},3]$-CH, where vertices in $C$ have 2-lists and all others have 3-lists. 

Suppose that the instance $(G, \{v_1,v_2,v_3\})$ is satisfiable for 3-color Pre-Ext; we shall prove that $G'$ is $[\{2,3\},3]$-choosable  if vertices of $C$ are the only vertices with 2-lists. Consider a 2-list assignment to vertices of $C$ in $G'$. Since $C$ is $[2,3]$-choosable (see Theorem~\ref{prop: 2-3-choos}) it has a list 3-coloring. If $v_1,v_2, v_3$ are colored with one or two colors, then this list coloring can be extended to  $G'$ by coloring $W$ with the third color. If they are colored using three colors then this coloring can be extended to $G'$ using the fact that $(G, \{v_1,v_2,v_3\})$ is satisfiable. 

Suppose conversely that $(G, \{v_1,v_2,v_3\})$ is not satisfiable and consider the following 2-list assignment for $C$: $L(v_1)=\{1,2\}$, $L(u_1)=\{2,3\}$, 
$L(v_2)=\{1,3\}$, $L(u_2)=\{1,2\}$, $L(v_3)=\{2,3\}$, $L(u_3)=\{1,3\}$. It is straightforward to verify that in any list coloring of $C$ $v_1, v_2$ and $v_3$ are colored using three colors and consequently it cannot be extended to the whole graph $G'$. This shows that in this case $G'$ is not choosable and concludes the proof of {\bf (i)}.

{\bf (ii)} and {\bf (iii)}.
Consider an instance  $[\{2,3\},3]$-CH and denote by $H$ the subgraph induced by vertices with 2-lists. Then the instance can be written $(G,H)$ where $G=(B\cup W, E)$ and $H=(B_H\cup W_H, E_H)$ with $B_H\subset B$ and $W_H\subset W$. Suppose that $|B_H\cup W_H|\leq 6$  and $H$ is not a $C_6$. 
Consider any list assignment with 2-lists on $B_H\cup W_H$ and 3-lists elsewhere. If $|B_H|\leq 2$ then it will be possible to color $B$ with a single color and $W$ with the two other colors; the same holds symmetrically if $|W_H|\leq 2$. In particular, this holds if $|B_H\cup W_H|\leq 5$, proving {\bf (iii)}. It remains to consider the case where $|B_H|=|W_H|=3$. Suppose first that there is in $H$ a vertex $u$ of degree  $d_H(u)\leq 1$; w.l.o.g. we can suppose $u\in W_H$. We can color $W\setminus \{u\}$ with one single color  since it contains only two vertices with 2-lists. Then it is possible to color  $(B\setminus B_H)\cup \Gamma (u)$ where $\Gamma(\cdot)$ denotes the neighborhood. Every vertex in $B_H\setminus (\Gamma(u)\cap B_H)$ has a monochromatic neighborhood and the same holds for $u$; consequently we can extend this list coloring to the whole graph.  Suppose now that every vertex in $H$ has a degree at least 2. Since 
$|B_H|=|W_H|=3$, $H$ contains a $C_6$ as a partial subgraph. Since $H$ is not a $C_6$, this cycle has a chord which means that $H$ contains a chocolate as a partial subgraph. Thus, $H$ is not $[2,3]$-choosable and $G$ is not choosable. As a conclusion,  to decide whether $G$ is $[\{2,3\},3]$-choosable we just have to detect whether $H$ contains a vertex of degree at most~1, which can be done in polynomial time. This concludes the proof of {\bf (ii)}.
\end{proof}

For complexity of choosability problems the natural class to consider is the class $\Pi_2^p$.
To our knowledge the previous results are among the first NP-completeness results in this area.
In~\cite{rubin}, it is shown that $\{2,3\}$-CH is $\Pi_2^p$-complete in bipartite graphs of maximum degree~4.
In~\cite{gutner}, the proof is adapted to show that $\{2,3\}$-CH is $\Pi_2^p$-complete in planar bipartite graphs of maximum degree~5. 
The proof uses only list assignments with a palette of seven colors and thus shows that  $[\{2,3\},7]$-CH is $\Pi_2^p$-complete in planar bipartite graphs of degree~5.
One can easily find list assignments with only four colors  and with exactly the same properties; as a consequence the  result effectively derived in~\cite{gutner} is that $[\{2,3\},4]$-CH is $\Pi_2^p$-complete in planar bipartite graphs of degree~5. In what follows, we  modify the reduction to obtain $\Pi_2^p$-completeness on subgrids using only three colors, an improvement of the previous hardness result:

\begin{theorem}\label{prop: 2,3-choos 4 col grids}
$\{2,3\}$-CH and, for any $k\geq 3$, $[\{2,3\},k]$-CH are $\Pi_2^p$-complete in subgrids, even if vertices of degree~4 have 3-lists and vertices of degree~2 have 2-lists.
\end{theorem}

\begin{proof}

Membership in $\Pi_2^p$ is clear, as a special case of the general choosability problem. As in the proof of~\cite{gutner}, our reduction uses the Restricted Planar Satisfiability (RPS), shown to be $\Pi_2^p$-complete in~\cite{gutner}. An instance ${\cal I}$ is an expression of the form $\forall U_1 \cdots \forall U_k \exists V_1 \cdots \exists V_r \Phi$, where $\Phi$ is a formula in conjunctive normal form over the variables $\{U_1, \ldots, U_s, V_1, \ldots, V_t\}$, each clause contains exactly three literals (either a variable or its opposite) corresponding to different variables, each variable occurs in at most three clauses, and the bipartite graph $\Gamma_{\Phi}$ with clauses in one part and variables in the other part, linking a clause to the three variables it involves, is planar. The question is whether the expression is true or not, which means that, for all truth assignment of variables $U_i$, $i=1, \ldots, s$ (called {\em $\forall$-variables}) there is a truth assignment of variables $V_j$, $j=1, \ldots, t$ (called {\em $\exists$-variables}) such that each clause contains at least one true literal.

\begin{figure}[h]
\begin{center}
\includegraphics[scale=0.35]{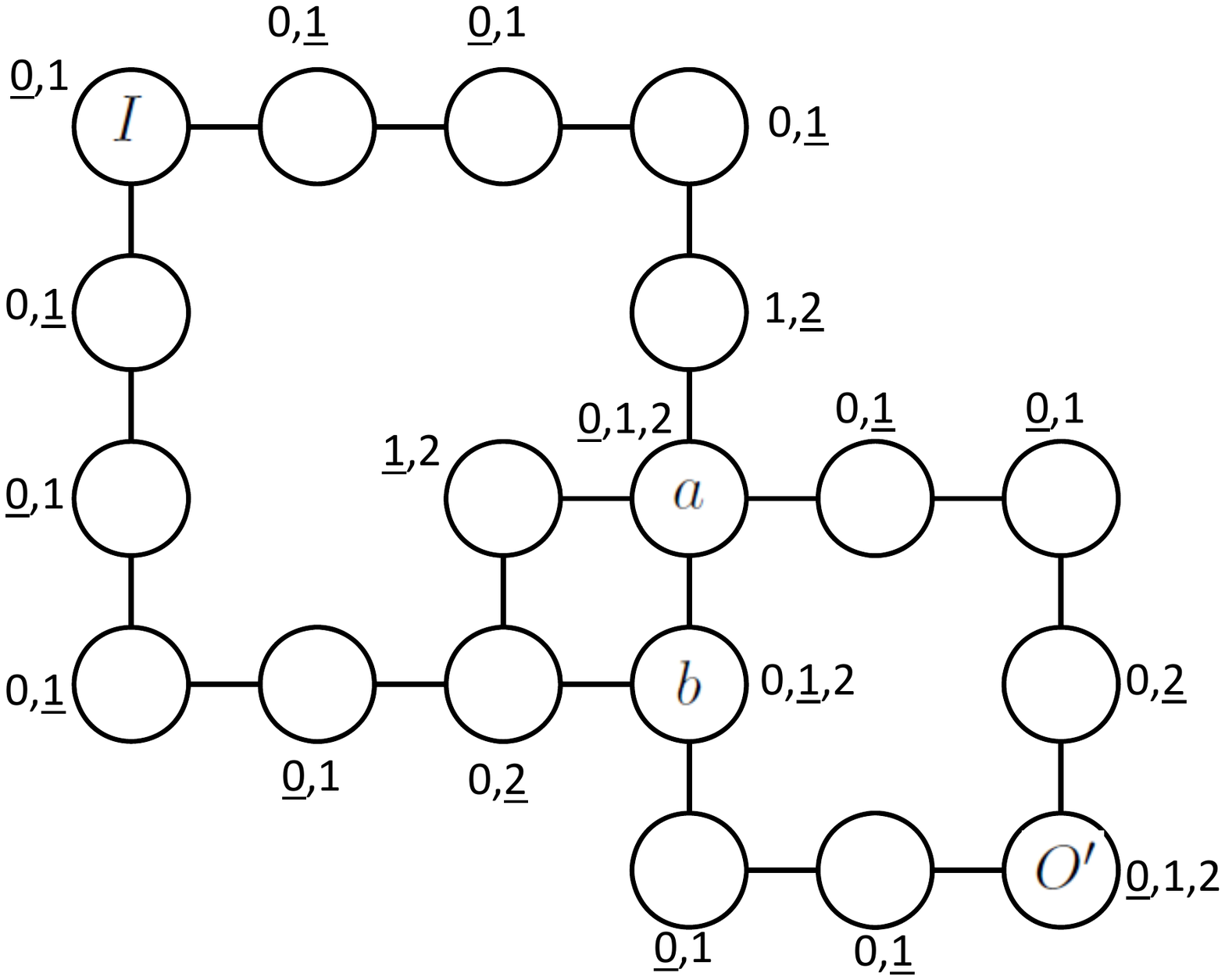}
\caption{The gadget $P_{1/2}$ with a list assignment for which imposing color $0$ at $I$ imposes color $0$ at $O'$.}
\label{fig:half-propagator}
\end{center}
\end{figure}

\begin{figure}[h]
\begin{center}
\includegraphics[scale=0.35]{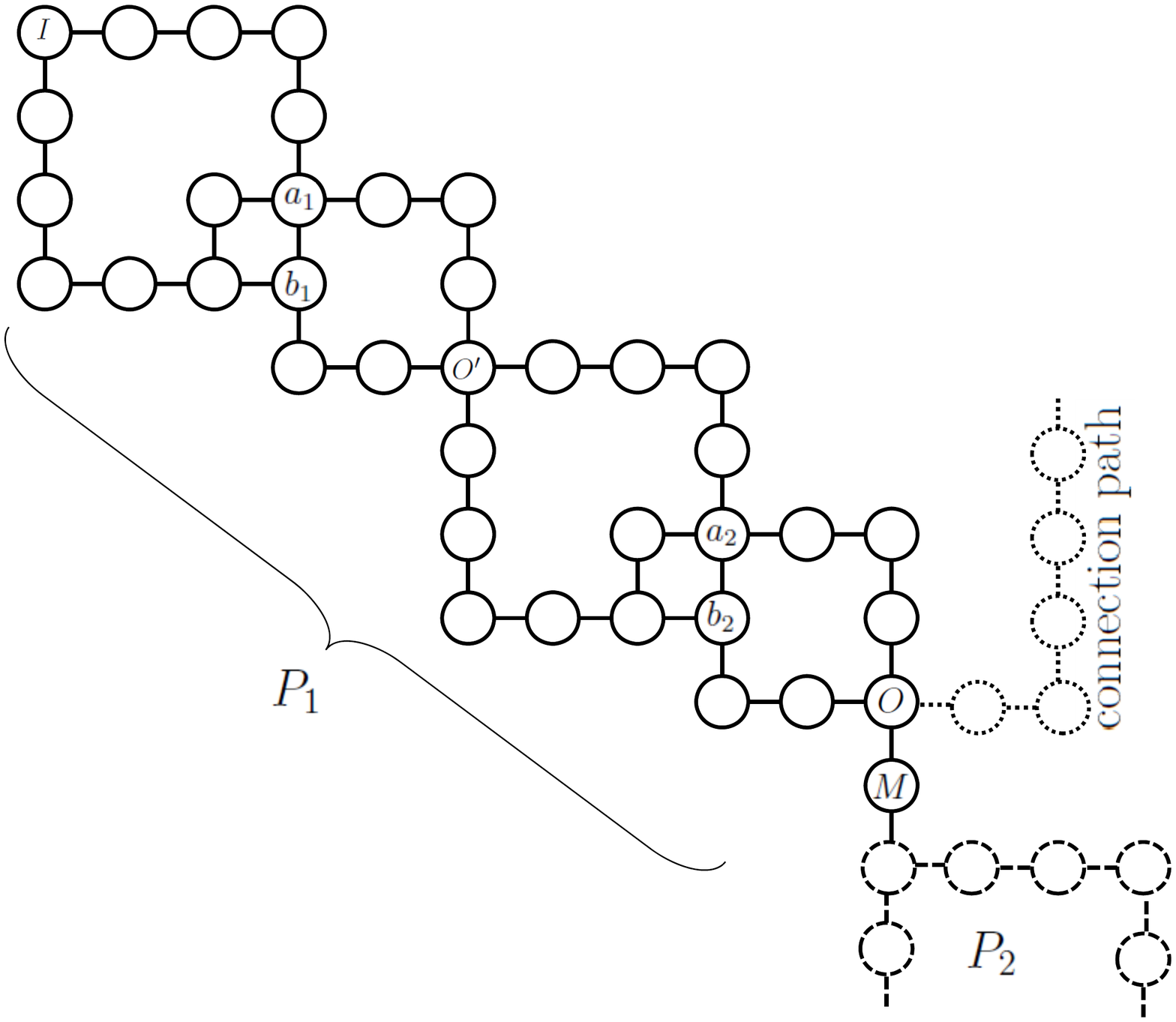}
\caption{Two gadgets $P_1$ and $P_2$ linked by a two-edge path.}
\label{fig:propagators}
\end{center}
\end{figure}

\begin{figure}[h]
\begin{center}
\includegraphics[scale=0.25]{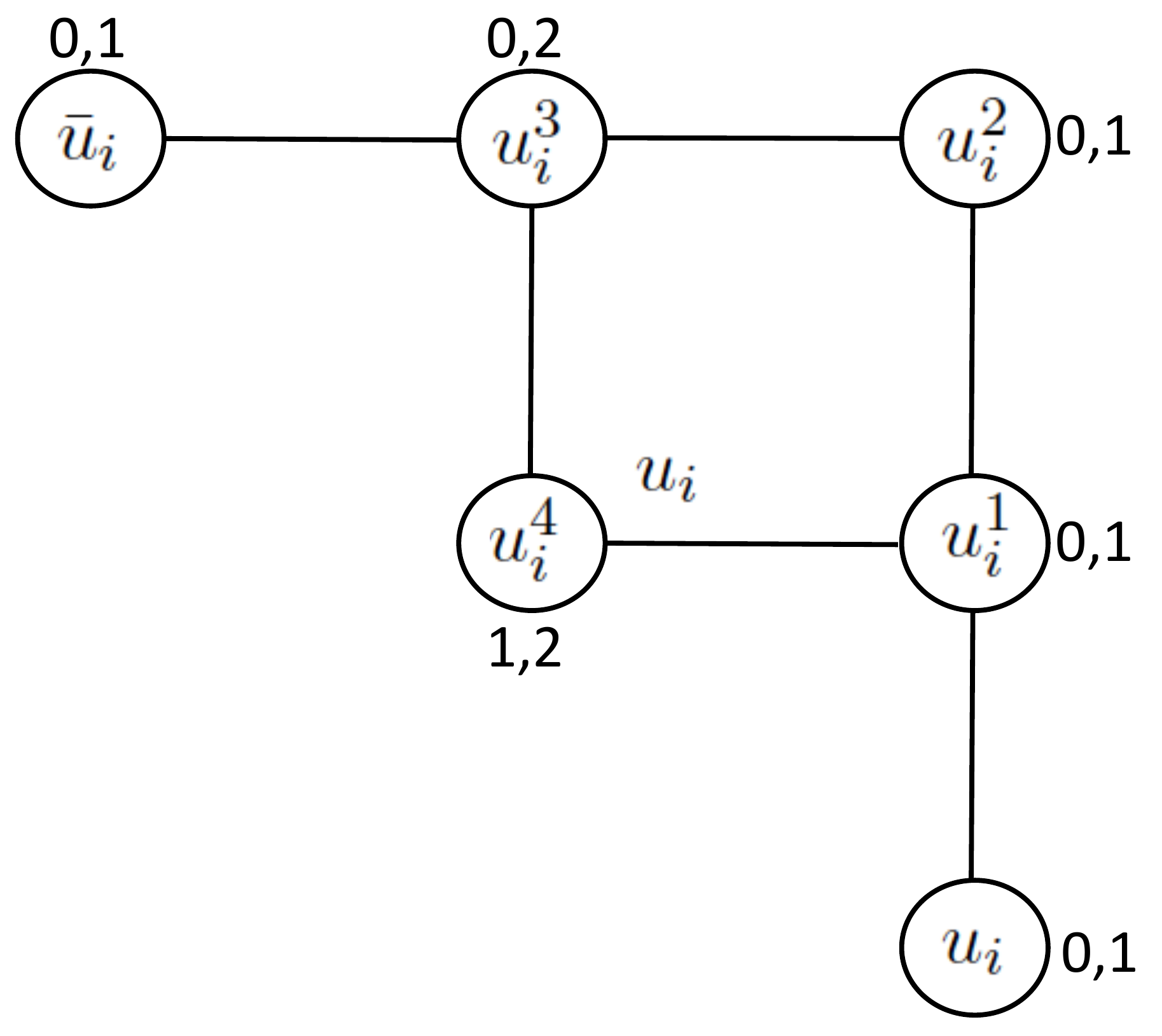}
\caption{The gadget associated with a $\forall$-variable $U_i$ with a 2-list assignment imposing color 1 to vertex $u_i$. }
\label{fig:forall}
\end{center}
\end{figure}

Given an instance ${\cal I}$ of RPS, we need to construct in polynomial time a subgrid $\widetilde{\Gamma}_{\Phi}=(V,E)$ and a function $f:V\rightarrow \{2,3\}$ such that   ${\cal I}$ is satisfiable if and only if $\widetilde{\Gamma}_{\Phi}$ is $[f,k]$-choosable, $k\geq 3$. 

We will now give a construction of $\widetilde{\Gamma}_{\Phi}=(V,E)$ which is different from those in~\cite{rubin} and~\cite{gutner}.

Without loss of generality we can assume that, in ${\cal I}$, every variable appears at least once in positive form and once in negative form and consequently, each literal has either 1 or 2 occurrences in different clauses of expression~$\Phi$.

We use gadgets which have properties  similar to those in~\cite{rubin} and~\cite{gutner}, but which can be embedded in a grid.  The main gadget, called, {\em half propagator} in~\cite{rubin} and~\cite{gutner} is replaced by the graph $P_{1/2}$ in Figure~\ref{fig:half-propagator}. Vertex $I$ is  the {\em Input} while $O'$ is  the {\em Output}. In $P_{1/2}$, we also specify two vertices $a,b$ as shown in Figure~\ref{fig:half-propagator}. 

A propagator $P$ consists of two gadgets $P_{1/2}$, glued by identifying  the Output of the first one and the Input of the second one. We then call $I$ the Input of the first $P_{1/2}$ and $O$ the Output of the second one, they form respectively the Input and the Output of the gadget $P$ (see Figure~\ref{fig:propagators}). We also define the function $f_P$ that associates value 3 with $O', O$ and the four vertices $a_1,b_1, a_2, b_2$. Two such gadgets $P_1, P_2$ may be connected by a two-edge path between the Output of $P_1$ and the Input of  $P_2$ (see Figure~\ref{fig:propagators}). 
We define $f_{P_{1/2}}(x)=3$ for $x\in\{a,b,O'\}$ and $f_{P_{1/2}}(x)=2$ for other vertices of $P_{1/2}$. 

We then consider the planar bipartite graph $\Gamma_{\Phi}=(V_\Phi, E_\Phi)$. Every vertex associated with a $\exists$-variable $V_j$ is replaced by two vertices $v_j,\bar v_j$ (one per literal) linked by an edge. Every vertex associated with a  $\forall$-variable $U_i$, $i=1, \ldots, s$ is replaced by six vertices $u_i,\bar u_i, u_i^1, u_i^2 , u_i^3, u_i^4$ with edges $u_i^1 u_i^2, u_i^2 u_i^3,u_i^3 u_i^4, u_i^4 u_i^1,$ $u_i u_i^1, \bar u_i u_i^3$ (see Figure~\ref{fig:forall}). 
These first two  gadgets are called {\em variable gadgets} with a specific vertex associated with each literal. Consider a literal $W$  and its related vertex $w$; if $W$ has a single occurrence in $\Phi$, then we combine $w$ with a gadget $P$ by identifying $w$ and $I$. If $W$ has two occurrences in the expression $\Phi$, then we combine   $w$ with two gadgets $P$ separated by a two-edge path  by identifying the vertex $w$ with the vertex $I$ of the first gadget $P$. We then denote by  $M$ the intermediate vertex between the two consecutive $P$'s (See Figure~\ref{fig:propagators}).   Each variable  is then associated with a subgrid with two or  three Output vertices (one by occurrence of the variable in $\Phi$) all situated on the external face of this subgrid and two Input vertices. Replacing in $\Gamma_{\Phi}$ the vertex associated with this variable by the related subgrid, it is possible to connect the three clause vertices linked to this variable in $\Gamma_{\Phi}$ to the related Output vertices without edge crossing. We denote by $E'_\Phi$ these edges that replace edges $E_\Phi$ of $\Gamma_\Phi$.  The resulting graph $\Gamma'_{\Phi}$ remains planar. Then it is possible to embed $\Gamma'_{\Phi}$ in a grid such that each edge in $E'_\Phi$ is replaced by a path called  {\em connection path}. This is possible since the degree of Output vertices is at most~4. Without loss of generality we can assume that connection paths contain at least two edges. Moreover the connection paths can be drawn in such a way that the resulting graph $\widetilde{\Gamma}_{\Phi}$ is a subgrid. Figure~\ref{fig:3-propagators} shows for instance the representation in the subgrid $\widetilde{\Gamma}_{\Phi}$ of a $\forall$-variable $U_i$ involved in three clauses $(u_1, u_3, \bar u_i), (\bar u_2, u_i, u_4), (\bar u_1, u_i, u_2), i\neq 1,2,3,4$. 

\begin{figure}[h]
\begin{center}

\includegraphics[scale=0.5]{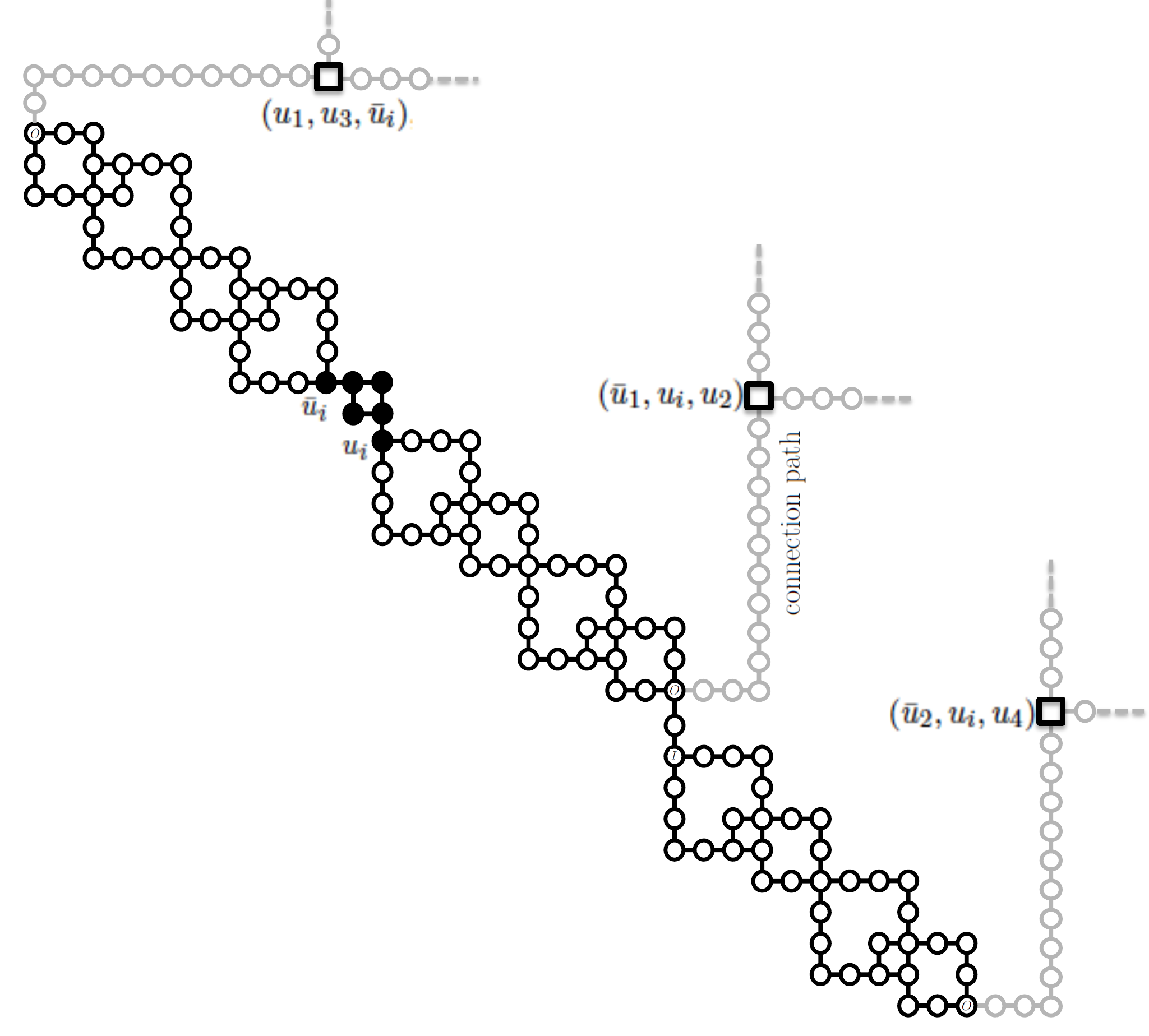}

\caption{Representation in $\widetilde{\Gamma}_{\Phi}$ of a $\forall$-variable $U_i$ involved in three clauses $(u_1, u_3, \bar u_i), (\bar u_2, u_i, u_4), (\bar u_1, u_i, u_2), i\neq 1,2,3,4$. The black subgraph (six black vertices) corresponds to the variable gadget associated with $U_i$ (as given in Figure~\ref{fig:forall}) and the three square vertices correspond to clause vertices.}
\label{fig:3-propagators}
\end{center}
\end{figure}

Function $f$ is obtained by associating the value~3 to vertices $a,b, O', O$ in gadgets $P$. Clause vertices are also associated with~3 and all other vertices get the value~2. In particular, the function $f_P$ is the restriction of $f$ to the related gadget $P$. Note that the only vertices of degree~4 in the subgrid $\widetilde{\Gamma}_{\Phi}$ are $a$ and $O'$ vertices of propagators $P$ and $O$ vertices of the first propagators $P_1$ associated with literals having two occurrences in  $\Phi$. All these vertices have 3-lists. All other vertices with 3-lists are of degree~3 and consequently vertices of degree~2 have 2-lists. There is no vertex of degree~1 or 0.
This completes the construction; it is performed in polynomial time. 

We now use the following Lemma.  For readability purposes its proof is given in the appendix. The arguments are similar to those of~\cite{rubin} and~\cite{gutner}. 

\begin{lemma}\label{lem_annex} \mbox{}\\
(i) if ${\cal I}$ is satisfiable, then $\widetilde{\Gamma}_{\Phi}$ is $f$-choosable.\\
(ii) if $\widetilde{\Gamma}_{\Phi}$ is $[f,3]$-choosable, then ${\cal I}$ is satisfiable.
\end{lemma}

We conclude the proof by observing that, as mentioned in Section~\ref{sec:notations}, $f$-choosability implies $[f,k]$-choosability for any $k$ and that, for any $k\geq 3$, $[f,k]$-choosability implies $[f,3]$-choosability.
\end{proof}

\begin{corollary}\label{cor:grid 2,3,5}
$\{2,3, 5\}$-CH and, for any $k\geq 5$, $[\{2,3, 5\},k]$-CH are $\Pi_2^p$-complete in grids.
\end{corollary}

\begin{proof}
Since both problems are in $\Pi_2^p$ we just need to show the completeness of $[\{2,3, 5\},k]$-CH for $k\geq 5$.
Starting from any instance  $({\cal S}, f_{\cal S})$, of $[\{2,3\},k]$-CH, where ${\cal S}$ is a subgrid, we embed in polynomial time ${\cal S}$ into a grid ${\cal G}$ and define $f_{\cal G}$ which extends $f_{\cal S}$ to $V({\cal G})$ by setting $f_{\cal G}(v)=5\ \forall v\in V({\cal G})\setminus V({\cal S})$. 
If ${\cal S}$ is $[f_{\cal S},k]$-choosable, then ${\cal G}$ is $[f_{\cal G},k]$-choosable: considering any list $k$-coloring of ${\cal S}$, we can  extend the coloring to vertices in $V({\cal G})\setminus V({\cal S})$. Indeed consider the vertices in $V({\cal G})\setminus V({\cal S})$ in any order;   since ${\cal G}$ is of maximum degree~4 and lists are of size~5 we can always find an available color in the list.

Conversely if  ${\cal G}$ is $[f_{\cal G},k]$-choosable, then its subgraph ${\cal S}$ is also $[f_{\cal S},k]$-choosable, which concludes the proof. 
\end{proof}

\subsection{3-choosability with bounded palette}\label{subsec: 3 choos}

Many results in choosability are concerned with subclasses of planar graphs (See, e.g.,~\cite{alta92,krato-tuza,bordeaux-3-col,thomassen}). Bipartite planar graphs are known to be 3-choosable~\cite{alta92}; this excludes the possibility to extend to 3-choosability hardness results known for $\{2,3\}$-choosability for these graphs. Triangle-free planar graphs are known to be 4-choosable~\cite{krato-tuza} and 3-colorable (reference (Grotzsch’s theorem, 1959)). In~\cite{gutner}, 3-CH is proved $\Pi_2^p$-complete in triangle-free planar graphs; however the proof involves 17 colors, which shows that $[3,17]$-CH is $\Pi_2^p$-complete in this class.  We will establish hardness results for $[3,k]$-CH in 3-colorable planar graphs for $k\geq 4$ and in triangle-free planar graphs for $k\geq 5$. Finally we settle the complexity of $[\ell,k]$-CH in bipartite graphs. 

The following technical lemma is a crucial instrument for deriving from Theorem~\ref{prop: 2,3-choos 4 col grids} hardness results for $[\ell,k]$-CH and in particular for $\ell=3$.  We remind that  criticality is defined in Section~\ref{sec:notations}.

\begin{lemma}\label{lem: ff'}
Let $H=(V_H,E_H)$ be a graph, $S_H\subset V_H$ a stable set in $H$ and $f_H: V_H \longrightarrow \mathbb{N}$ a function. We suppose:
\begin{enumerate}
\item\label{item: lemff'1} $H$ is  $([f_H,k], S_H)$-critical for some $k\geq \max\left (\max\limits_{v\in V_H}(f_H(v)); \max\limits_{v\in S_H}(f_H(v)+1)\right )$.
\item\label{item: lemff'2} Every list assignment $L_H$ with a palette of $k$ colors satisfying the following conditions is feasible:  $\forall v\in V_H\setminus S_H, |L_H(v)|=f_H(v)$ and  $\forall v \in S_H, L_H(v)=\{c\}$  for some color~$c$.
\end{enumerate}
Given a graph $G=(V,E)$, $f: V \longrightarrow \mathbb{N}$, $v_0\in V$, $\max_{v\in V}(f(v))\leq k$, $f(v_0)\leq k-1$, we define $G'=(V_H\cup V, E')$ obtained by adding $H$ to $G$ and linking every vertex in $S_H$ with $v_0$ and $f': V_H\cup V \longrightarrow \mathbb{N}$, defined by $f'(v)=f_H(v)+1\ \forall v\in S_H, f'(v_0)=f(v_0)+1, f'(v) =f_H(v),\ \forall v\in V_H\setminus S_H$ and $f'(v)=f(v)\ \forall v\in V\setminus \{v_0\}$.

Then $G'$ is $[f',k]$-choosable if and only if $G$ is $[f,k]$-choosable.
\end{lemma}

\begin{proof}
Note that under the assumptions we have $\max (f')\leq k$. Suppose first that $G$ is $[f,k]$-choosable and consider an $f'$-list assignment $L'$ for $G'$. We consider two cases:

{\underline{Case 1.}} Suppose $\cap_{v\in S_H}L'(v)\neq \emptyset$. We color $S_H$ with a single color $c\in \cap_{v\in S_H}L'(v)$. Using assumption~\ref{item: lemff'2}, this coloring can be extended to  a feasible list coloring in $V_H$. Removing $c$ from $L'(v_0)$ or any color if $c\notin L(v_0)$, we can extend the list coloring to $V$ since $G$ is $[f,k]$-choosable.

{\underline{Case 2.}} If $\cap_{v\in S_H}L'(v)= \emptyset$, then we first consider a list coloring of $G$; let $c$ be the color of $v_0$ and remove it from lists in $S_H$. Since at least one list in $S_H$ did not change; by using the fact that $H$ is $([f_H,k], S_H)$-critical, we can find a list $k$-coloring of $V_H$ without using $c$ in $S_H$, which shows that $G'$ is $[f',k]$-choosable. 

Suppose now that 
$G'$ is $[f',k]$-choosable and consider an $f$-list assignment $L$ of $G$. By definition, since $H$ is $([f_H,k], S_H)$-critical, there is an infeasible $f_H$-list assignment $L_H$  such that $|\cup_{v\in S_H}L_H(v)|\leq k-1$. Let $c\in \{1, \ldots, k\}$ such that $c\notin \cup_{v\in S_H}L_H(v)$; w.l.o.g.,  since $f(v_0)\leq k-1$, by performing a circular permutation on colors we can suppose that $c\notin L(v_0)$. We then add $c$ in $L_H(v), v\in S_H$ and in $L(v_0)$ to get an $f'$-list assignment $L'$ for $G'$. $G'$ being $[f',k]$-choosable, there is a feasible list coloring and by property of the assignment $L_H$, at least one element in $S_H$ gets color $c$. Consequently the list coloring of $G'$ with respect to $L'$ defines a list coloring of $G$ with respect to $L$, which concludes the proof.  
\end{proof}

 We remind that an {\em odd hole} in a graph $G$ is an induced cycle $C_{2i+1}, i\geq 2$.

\begin{theorem}\label{prop: 3-choos 4 col}
$[3,4]$-CH is $\Pi_2^p$-complete in 3-colorable odd hole-free planar graphs of maximum degree~6.
\end{theorem}
\begin{proof}
Consider a $K_4$ on vertex set $\{v_1,v_2,v_3,v_4\}$ and let $D$ be the diamond (or $K_{1,1,2}$) obtained by removing edge $v_1v_4$. We denote by $f_D$ the function that associates to each vertex of $D$ its degree. We then have the following claim:

\begin{claim}\label{claim:diamond}
	$D$ is $f_D$-choosable.
\end{claim}
\begin{proof}(of Claim~\ref{claim:diamond})
Suppose any $f_D$-list assignment with $k$ colors ($k\geq 3$) and without loss of generality we suppose that $L(v_1)=\{1,2\}$ and $L(v_4)=\{a,b\}$. Suppose there is no list coloring with $v_4$ colored $a$. Then we claim that $L(v_2)\cup L(v_3)\setminus \{a,1,2\}=\emptyset$. In the opposite case indeed, $v_2$ or $v_3$ could be colored with a color $c\neq a,1,2$ and then the second vertex of degree~3 in $D$ could be colored since at least one color would be available in its list. Finally $v_1$ could be colored since at most one of its neighbors would be colored 1 or 2. Since $|L(v_2)\cup L(v_3)|\geq 3$ and $L(v_2)\cup L(v_3)\setminus \{a,1,2\}=\emptyset$, we have $a\neq 1,2$ and$L(v_2)\cup L(v_3)=\{a,1,2\}$. Consequently, $L(v_2)= L(v_3)=\{a,1,2\}$. In this case we can color $v_4$ with $b$, $v_2$ with $a\neq b$, $v_3$ with 1 or 2 different from $b$ and thus, at least one color is still available for coloring $v_1$, which concludes proof of the claim.
\end{proof}

\begin{figure}[h]
\begin{center}
\includegraphics[scale=0.4]{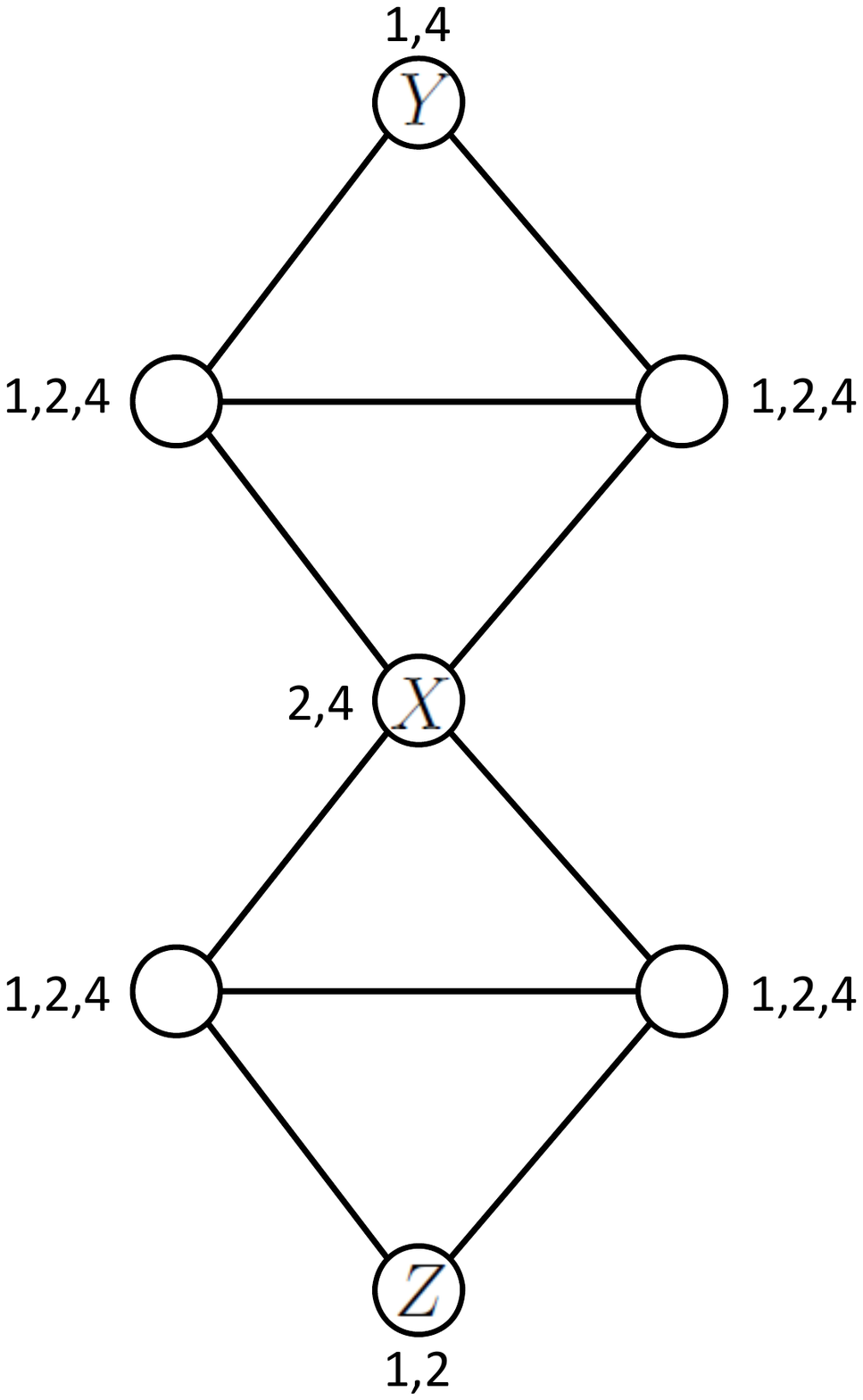}
\caption{The gadget $H$ with an infeasible $f_H$-list assignment.}
\label{fig: 3-choos 4 col}
\end{center}
\end{figure}

Consider the gadget $H$ in Figure~\ref{fig: 3-choos 4 col}; it is obtained from two diamonds by merging one vertex of degree~2 in each diamond into a new vertex $X$. We call $Y,Z$ the two remaining vertices of degree~2.
We apply Lemma~\ref{lem: ff'} to $H$ with $S_H=\{X,Y,Z\}$, $f_H(v)= 3, v\notin S_H$,  $f_H(v)=2, v\in S_H$ and $k=4$. To show that $H$ is  $([f_H,4], S_H)$-critical we consider the $f_H$-list assignment in Figure~\ref{fig: 3-choos 4 col}. Color 3 does not appear in $L_H(X)\cup L_H(Y)\cup L_H(Z)$ and it is straightforward to verify that the list assignment is infeasible. Suppose now that $X$ has a 3-list while $Y$ and $Z$ have a 2-list with colors in $\{1,2,3,4\}$. Suppose, without loss of generality, that $L_H(Y)= \{1,2\}$ and $L_H(X)= \{a,b,c\}$. Then, using Claim~\ref{claim:diamond} in the diamond $D_Y$  built on $X$ and $Y$ with $L(X)=\{a,b\}$ we know that at least one color, say $a$, can be used for $X$ in a feasible coloring of this diamond and then, using similarly Claim~\ref{claim:diamond} in $D_Y$ with $L(X)=\{b,c\}$, another color between $b,c$ for $X$, say $b$, allows to color $D_Y$. Assigning to $X$ the list $L(X)=\{a,b\}$ and applying  Claim~\ref{claim:diamond} in the diamond $D_Z$ built on $X$ and $Z$, it is possible to list color $D_Z$ using, for $X$, one of the two selected colors $a,b$ and by assumption on these colors it is possible to extend the list coloring to the whole gadget $H$.  Suppose now that $|L_H(Y)|=3$, while $|L_H(X)|=|L_H(Z)|=2$. Using Claim~\ref{claim:diamond} it is possible to find a feasible list coloring of the diamond $D_Z$ in $H$ and then the three last vertices of $H$ can be colored greedily, $Y$ being the last one. By symmetry, if $|L_H(Z)|=3$, while $|L_H(X)|=|L_H(Y)|=2$ it is also possible to list color $H$. This concludes that $H$ is  $([f_H,4], S_H)$-critical. To verify Assumption~\ref{item: lemff'2} of Lemma~\ref{lem: ff'} suppose that $X,Y,Z$ are colored with the same color $c$ and other vertices in $V_H\setminus\{X,Y,Z\}$ have 3-lists. Removing $c$ from their lists and removing $X,Y,Z$ from $H$ we get two independent edges (forming a $2K_2$) with lists of size at least~2. Since $2K_2$ is 2-choosable this list coloring of $X,Y,Z$ can be extended to $H$. So, Lemma~\ref{lem: ff'} applies.

Consider now a 
 subgrid $G$, instance of $[\{2,3\},4]$-CH and denote by $V^2, V^3$ the vertices of $G$ with lists of size 2 and 3, respectively. Using Theorem~\ref{prop: 2,3-choos 4 col grids} we can assume that $V^2\neq \emptyset$ and the maximum degree in $G$ of vertices in $V^2$ is~3. For every vertex $v\in V^2$, we add a gadget $H_v$ isomorphic to $H$  and link the related vertices $X_v,Y_v,Z_v$ to $v$. Let $\widetilde{G}$ be the resulting graph. $\widetilde{G}$ is planar and of maximum degree~6 (this is obtained from vertices in $V^2$ of degree~3 in $G$). Note that $\widetilde{G}$ is 3-colorable: color $G$ with colors 1,2 and, for every gadget $H_v$ color vertices $X,Y,Z$ with color 3 and the other vertices with colors 1,2. Moreover the only elementary odd cycles in $\widetilde{G}$ can be found in subgraphs induced by a vertex $v$ in $V^2$ and its related gadget $H_v$. Such a graph, isomorphic to a gadget $H$ with an additional vertex linked to $X,Y,Z$, has no hole: the only odd cycles of length 5 or 7 have chords. This shows that  $\widetilde{G}$ is odd hole-free. 
 
 Using Lemma~\ref{lem: ff'} for every $v\in V^2$, we show that
$\widetilde{G}$ is $[3,4]$-choosable if and only if $G$ is $[\{2,3\},4]$-choosable. Using Theorem~\ref{prop: 2,3-choos 4 col grids} we conclude the proof.
\end{proof}

Note that graph $W_3$ in~\cite{gutner} is 3-colorable and  $([3,5],\{v_0\})$-critical for a vertex $v_0$ of degree~6. Using this graph as $H$ is the previous proof allows to derive hardness of $[3,5]$-CH in 3-colorable planar graphs of maximum degree~7. Our proof allows to reduce the palette to~4, which is optimal since these graphs are $[3,3]$-choosable.

We remind that $[3,17]$-CH is known to be $\Pi_2^p$-complete in triangle-free planar graphs~\cite{gutner}. We propose a construction allowing to reduce the palette size to~5. This leaves open the status of $[3,4]$-choosability in this class.
\begin{theorem}\label{prop: 3-choos 5 col planar}
$[3,5]$-CH is $\Pi_2^p$-complete in triangle-free planar graphs of maximum degree at most~12.
\end{theorem}

\begin{proof}

\begin{figure}[h]
\begin{center}
\includegraphics[scale=0.4]{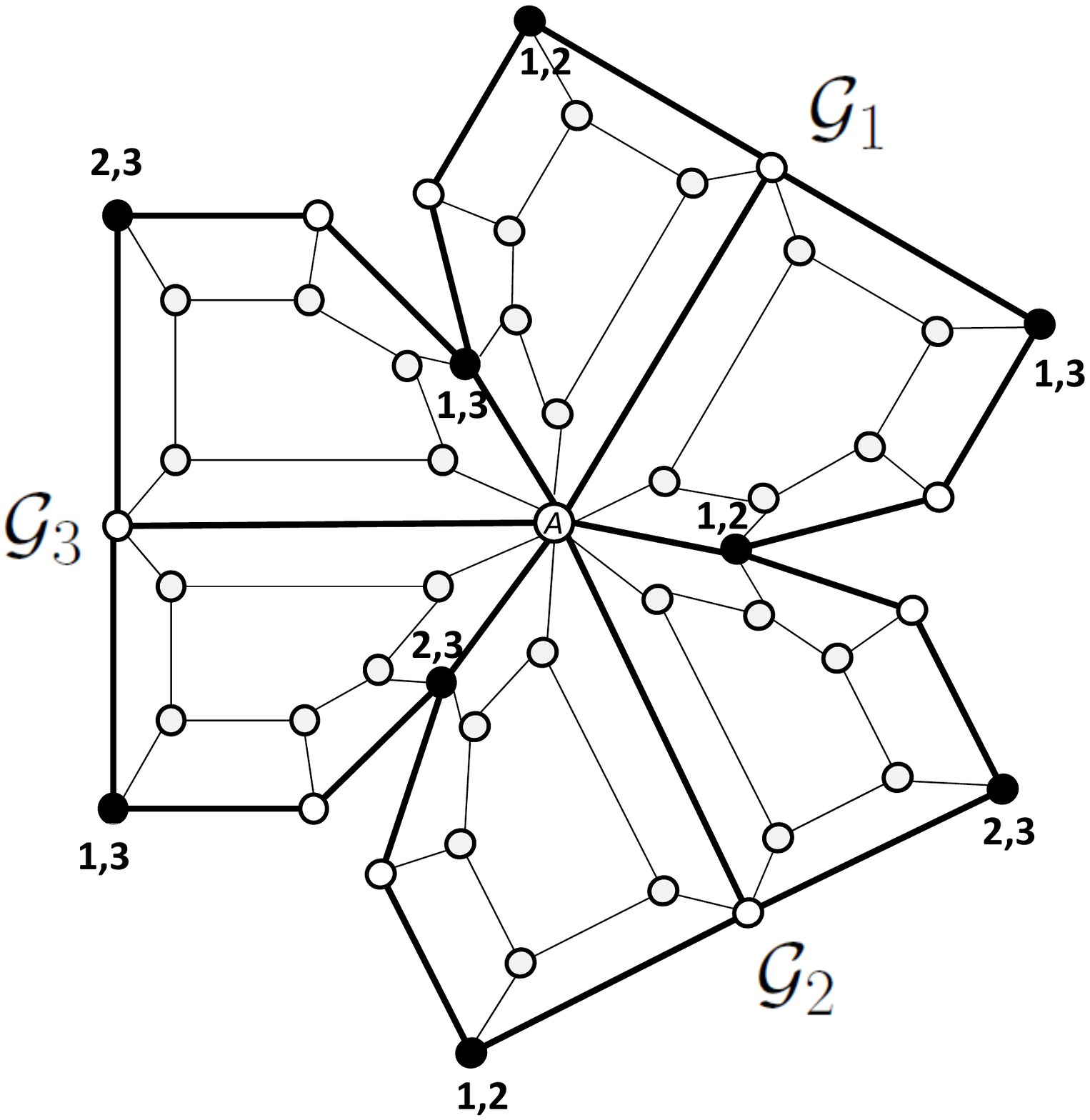}
\caption{The gadget ${\cal G}$.}
\label{fig: 3-5-fig-gen}
\end{center}
\end{figure}

\begin{figure}[h]
\begin{center}
\includegraphics[scale=0.4]{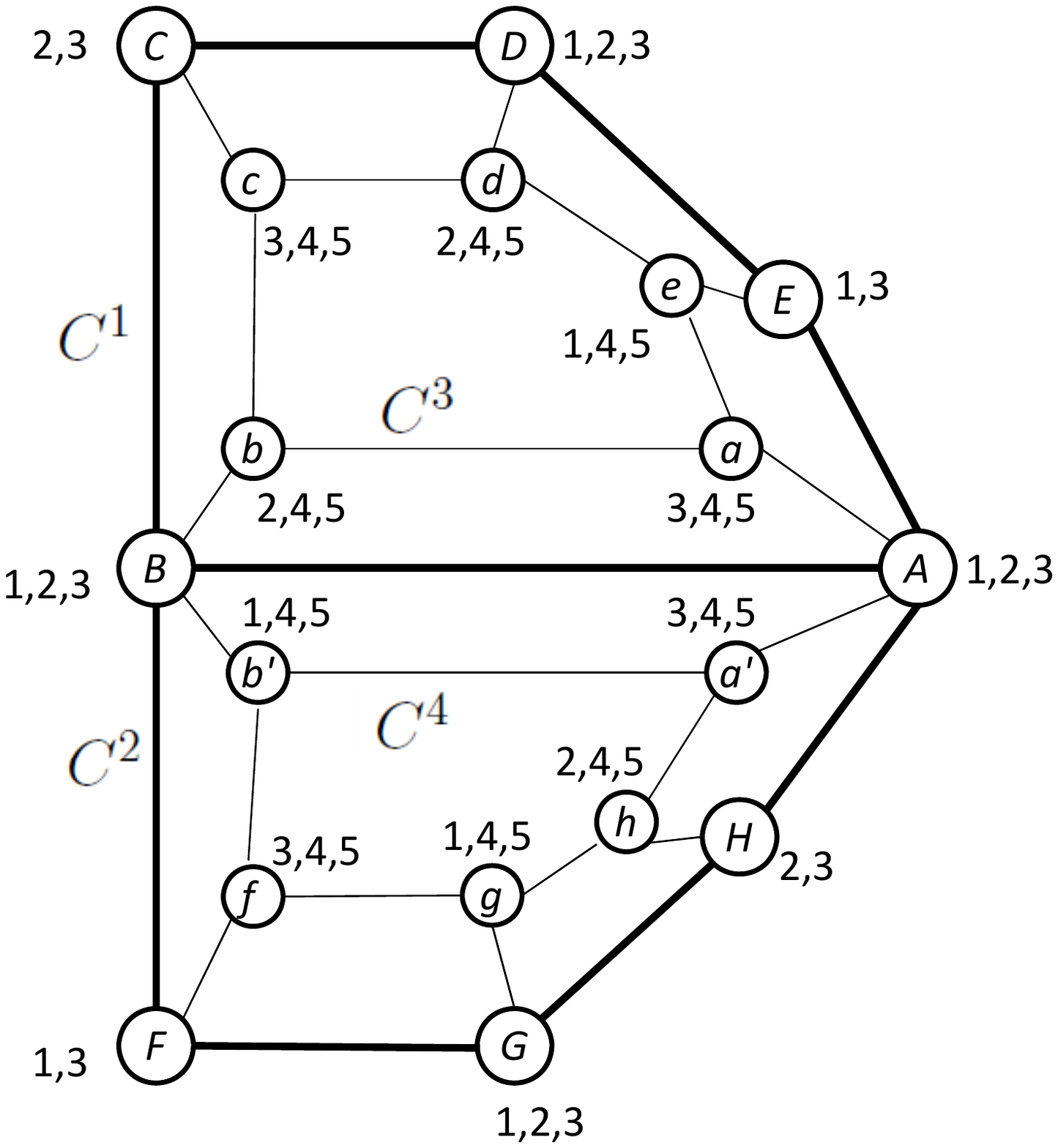}
\caption{The gadget ${\cal G}_3$ with  $f_3$-list assignment forbidding color 3 for $A$.}
\label{fig: 3-5-fig-part}
\end{center}
\end{figure}

Here, we consider the gadget ${\cal G}$ in Figure~\ref{fig: 3-5-fig-gen} with central vertex $A$, of degree 12 and all other vertices of degree at most~5. It is made of thee isomorphic gadgets ${\cal G}_1, {\cal G}_2, {\cal G}_3$ pairwise glued by identifying two vertices. Figure~\ref{fig: 3-5-fig-part} represents ${\cal G}_3$ with a list assignment making color 3 forbidden at vertex $A$.
${\cal G}_3$ is made of two $C_5$'s, $C^1= (A,B,C,D,E)$ and $C^2=(A,B,F,G,H)$, glued by the edge $AB$ with two other $C_5$'s, namely $C^3=(a,b,c,d,e)$ linked to $C^1$ by a matching of size 5 and $C^4=(a',b',f,g,h)$ linked to $C^2$ by a matching of size 5. We define a function $f_3: V({{\cal G}_3})\rightarrow \{2,3\}$ associating 2 to $C,E,F,H$ and 3 to all other vertices. $f_3$ can be immediately extended to $V({\cal G})$ and called $f$: ${\cal G}$ is planar and vertices on the external face are alternately associated with 2 and 3 with vertices in two different ${\cal G}_i$ associated with~2. They form a set  $S$ which corresponds to black vertices in Figure~\ref{fig: 3-5-fig-gen}. Note that $|S|=9$.  

In what follows, we show that ${\cal G}$ is $([f,5], S)$-critical. It is straightforward to verify that with the $f_3$-list assignment for ${\cal G}_3$  in Figure~\ref{fig: 3-5-fig-part}, color 3 is forbidden for vertex $A$. Hence, by a permutation of colors we can extend this list assignment to ${\cal G}$ such that color 2 will be forbidden for $A$ (using lists in ${\cal G}_2$) as well as color 1 (using lists in ${\cal G}_1$). Figure~\ref{fig: 3-5-fig-gen} shows the related lists for vertices in $S$. 
We now need to show that changing $f(s_0)$ from 2 to 3,  for one vertex $s_0$ in $S$  and denoting by $f'$ the new function, ${\cal G}$ becomes $[f',5]$-choosable. Without loss of generality we can consider that $s_0$ is a vertex of ${\cal G}_3$ and moreover, by symmetry, that it is either $F$ or $H$. 

We first need to point out some choosability properties of a $C_5$. It is $\{2,3\}$-choosable if at least one vertex has a 3-list. If all other vertices have 2-lists, then it may happen that only one single list coloring is possible. If three consecutive vertices have their color fixed and the two other vertices have 3-lists of possible colors, then there are at least two different list colorings of the whole $C_5$. Finally, the $C_5$ is not 2-choosable but the only infeasible 2-list assignment consists of five identical 2-lists. Consider now two $C_5$'s say $C^\alpha$ and $C^\beta$ linked by a matching of size~5 (say $C^\beta$ inside $C^\alpha$), assign 3-lists to vertices in $C^\beta$ and $\{2,3\}-lists$ to vertices in $C^\alpha$. If $C^\alpha$ has two different list colorings, then it will be possible to extend it on vertices of $C^\beta$. Consider indeed a first list coloring of $C^\alpha$ that cannot be extended on $C^\beta$. Necessarily, after removing from the list of every vertex $v$ of $C^\beta$ the color of its unique neighbor in $C^\alpha$, all vertices in $C^\beta$ get the same list with two colors $c_1$ and $c_2$ that are not used in the coloring of $C^\alpha$. Since at least three different colors were needed to color $C^\alpha$, $c_1$ and $c_2$ are the only colors appearing in all 3-lists in $C^\beta$. Hence, considering the second possible list coloring of $C^\alpha$,  and removing from each list in $C^\beta$ the color of its neighbor in $C^\alpha$ (if it is in its list), at least two colors are still available for every vertex in $C^\beta$ and we cannot have five identical 2-lists. Thus, $C^\beta$ can be colored.

Consider now  ${\cal G}_3$ with any $f_3$-list assignment $L_3$, let $L_3(A)=\{i,j,k\}$. If there is no list coloring assigning color $i$ to $A$, then we show that there is at least one list coloring with $A$ colored $j$ and one with $A$ colored $k$. Moreover in each case, the color of $E$ can be arbitrarily chosen in $L_3(E)$ provided it is compatible with the color in $A$. Indeed, after coloring $A$ with $i$, there are still at least two possibilities for coloring  $C^1$ and only one blocks the coloring of the internal $C^3$. We call $col^1$ this coloring of $C^1$. This imposes the color of $B$ and then there is only one way to color $C^2$, called $col^2$, which blocks the coloring of the internal $C^4$. Hence, choosing the color of $A$ among $\{k,j\}$ and the color of $E$, compatible with the color of $A$ leads to an alternative way to color $C^1$ and $C^2$, different from $col^1$ and $col^2$ and consequently in each case the list coloring of the whole graph ${\cal G}_3$ can be completed. 

By a similar argument, we also see that if $i\notin L_3(H)$ then there is a list coloring of ${\cal G}_3$ with $A$ colored $i$ and $E$ colored with an arbitrary color of $L_3(E)$ not equal to $i$.

We are now ready to show that ${\cal G}$ is $[f',5]$-choosable. Consider any $f'$-list assignment, choose for $A$ a color that is neither forbidden in ${\cal G}_1$ nor in ${\cal G}_2$. 

{\em {\underline{Case 1:}}} $s_0=F$. Using the previous remarks, it is then possible to extend the coloring of $A$ in a list coloring of ${\cal G}_1\cup {\cal G}_2$. We then have for ${\cal G}_3$ the color of $A, E$ and $H$ fixed. There are still at least two possibilities for coloring $C^1$ with at least one that can be extended to $C^3$. Then, even if the color of $B$ is fixed, both $F$ and $G$ having 3-lists, there are at least two possibilities for coloring $C^2$ with at least one that can be extended to $C^4$.

{\em {\underline{Case 2:}}} $s_0=H$. Remove the color of $A$ from $L'(H)$ and apply the same arguments as previously.
This concludes that ${\cal G}$ is $([f,5], S)$-critical. 

The end of the proof is now similar to the proof of Theorem~\ref{prop: 3-choos 4 col}.
Consider a 
 subgrid $G$ instance of $[\{2,3\},5]$-CH, and denote by $V^2, V^3$ the vertices of $G$ with lists of size 2 and 3, respectively. Using Theorem~\ref{prop: 2,3-choos 4 col grids} we assume that $V^2\neq \emptyset$ and the maximum degree of vertices in $V^2$ is~3.
 
  For every vertex in $v\in V^2$, we add a gadget ${\cal G}_v$ isomorphic to ${\cal G}$  and link the related vertices in $S$ to $v$. Let $\widetilde{G}$ be the resulting graph. It is planar, triangle-free and of maximum degree at most 12 since $|S|=9$ and the maximum degree in $G$ of vertices in $V^2$ is~3. In $\widetilde{G}$, vertices of degree~12 are vertices in $V^2$ of degree~3 in $G$ as well as central vertices $A$ in gadgets~${\cal G}$.
  
   Using Lemma~\ref{lem: ff'} for every $v\in V^2$, we show that
$\widetilde{G}$ is $[3,5]$-choosable if and only if $G$ is $[\{2,3\},5]$-choosable. Using Theorem~\ref{prop: 2,3-choos 4 col grids} we conclude the proof. 
\end{proof}

\begin{remark}
In~\cite{gutner}, by using an auxiliary graph $W_2$, a triangle-free planar graph with 164 vertices which is not 3-choosable was constructed. A crucial module in this construction was the occurrence of two $C_5$'s linked by a matching. Our gadget ${\cal G}$ is based on the same module but the modules are combined in a different way in order to limit the number of colors used to $k=5$. This allows us to exhibit a triangle-free planar graph with only 148 vertices which is not $[3,5]$-choosable and hence not 3-choosable: take 3 copies of ${\cal G}$ and link a new vertex to all black vertices. 

An open question is the existence of a triangle-free planar graph that is not $[3,4]$-choosable.
\end{remark}

We conclude this part by a complete characterization  of the complexity of $[\ell,k]$-CH for bipartite graphs. 
For any $\ell\geq 3$, $\ell$-CH is known to be $\Pi_2^p$-complete in bipartite graphs~\cite{gutner-tarsi}.
More precisely, this result is obtained by reducing $[\{2,3\},k]$-CH to $[3,k+3]$-CH and then, for any $\ell\geq 3$, $[\ell,k]$-CH to $[\ell,k+\ell+1]$-CH. Using Theorem~\ref{prop: 2,3-choos 4 col grids}, it shows that $[\ell,\frac{\ell(\ell+1)}{2}]$-CH is  $\Pi_2^p$-complete in bipartite graphs. 

As pointed out in~\cite{cogis}, every bipartite graph is $[3,4]$-choosable and moreover for $k\leq 2(\ell -1)$ every bipartite graph is $[\ell,k]$-choosable. The problem turns out to be hard in bipartite graphs with a fifth color. More generally we have:

\begin{proposition}\label{prop: l-choos k col bip}
$[\ell,k]$-CH in bipartite graphs is:\\
(1) trivial if $2\leq k<2\ell-1$ (always choosable)\\
(2) polynomial for $k=2$ and $\ell \geq 3$,\\
(3)  $\Pi_2^p$-complete for $\ell\geq 3$ and $k\geq 2\ell-1$, this holds in particular for $[3,5]$-CH. 
\end{proposition}
\begin{proof}
(1) is already shown in~\cite{cogis} (note that for $k =2$ and $\ell =2$ it just states that the graph is bipartite) and (2) is a consequence of Theorem~\ref{prop: 2-3-choos}.

We only need to prove (3).
Membership to $\Pi_2^p$ is already established.  
We first remind that instances of $[\ell,k]$-CH are also instances of $[\ell,k']$-CH for $k'\geq k$ and consequently if $[\ell,k]$-CH is $\Pi_2^p$-complete, so does $[\ell,k']$-CH for $k'\geq k$. So, to prove (3) we assume $k=2\ell-1$ and denote by ${\cal K}$ the palette.

We first claim that, for any $\ell\geq 3$ and every $k\geq 2\ell-1$, the gadget $H=(B\cup W, E)=K_{\binom{2\ell-2}{\ell},\binom{2\ell-2}{\ell-1}}$ is $([(\ell,\ell-1),2\ell-1],W)$-critical: suppose indeed that $\ell$-lists of  vertices in $B$ describe all sets of $\ell$ colors among $\{1, \ldots, 2\ell-2\}$, and $(\ell-1)$-lists of vertices in $W$ describe all sets of $(\ell-1)$ colors among $\{1, \ldots, 2\ell-2\}$. Then at least $(\ell-1)$ colors are needed to color vertices in $B$ and consequently at least one $W$-vertex cannot be list colored. 

Moreover suppose that all vertices in $B$ are assigned to $\ell$-lists with colors in $\{1, \ldots, 2\ell-1\}$, while vertices of $W$ but one are assigned to $(\ell-1)$-lists and the last one has an $\ell$-list.  We will color $B$ with only $(\ell-1)$ colors. A set $S$ of $(\ell-1)$ colors (called a $(\ell-1)$-set) cannot color $B$ if and only if at least one of the lists of vertices in $B$ is ${\cal K}\setminus S$. So, at most $\binom{2\ell-2}{\ell}$ $(l-1)$-sets of colors do not allow to color all vertices in $B$ and consequently at least $\binom{2\ell-1}{\ell-1} -\binom{2\ell-2}{\ell}=\binom{2\ell-1}{\ell} -\binom{2\ell-2}{\ell}=\binom{2\ell-2}{\ell-1}$ $(l-1)$-sets of colors allow to color $B$. We choose one of them that is not one of the $\binom{2\ell-2}{\ell-1}-1$ $(\ell-1)$-lists of vertices in $W$. Now, $B$ being colored with these $(\ell-1)$ colors, it is possible to extend the list coloring to $W$: vertices with a $(\ell-1)$-list have at least one available color and the last vertex with a $\ell$-list has as well at least one available color.

We will apply Lemma~\ref{lem: ff'} with $H=(B\cup W, E)=K_{\binom{2\ell-2}{\ell},\binom{2\ell-2}{\ell-1}}$ and $S_H=W$. The  first assumption has just been stated. The second assumption also trivially applies: if we color $W$ with a single color, then ($\ell\geq 2$) we always can complete the list coloring to $B$. 

Consider now a bipartite graph $G=(V,E)$ instance of $\{2,3\}$-CH and call $f_G$ the function, with values 2 or 3, assigning to every vertex a list size. For any vertex $v_0$ we call $G'$ the graph obtained from $G$ by adding a copy of $H$ and connecting $v_0$ to all vertices in $W$. $G'$ is bipartite since $G$ is bipartite. We extend $f_G$ to $G'$ by setting $f_G(v)=f'_{G'}(v), v\in V,  v\neq v_0$,  $f'_{G'}(v_0)= f_G(v_0)+1$ and $f'_{G'}(v)=\ell$ for $v\in B\cup W$. Using Lemma~\ref{lem: ff'} we know that $G$ is $[f,k]$-choosable if and only if $G'$ is $[f',k]$-choosable. 
Repeating $(\ell-f(v_0))$-times this construction for every vertex $v_0\in V$ we construct in polynomial time a bipartite graph $\widetilde{G}$ that is $[\ell,k]$-choosable if and only if $G$ is $[f,k]$-choosable. 
Since $\{2,3\}$-CH is $\Pi_2^p$-complete in bipartite graphs~\cite{gutner} we conclude that $[\ell,k]$-CH  
is $\Pi_2^p$-complete in bipartite graphs. This concludes the proof.
\end{proof}

\section{Some remarks about complexity of choosability}\label{sec:remarks}

We conclude this work with a few remarks on the complexity of choosability. We first give evidences that the complexities of list coloring and choosability are independent by exhibiting classes of graphs where one of them is polynomial while the other one is hard. Then we discuss the relative complexity of $[\ell, k]$-choosability and $[\ell, k+1]$-choosability and give examples where adding a color in the palette makes the problem harder or easier.

\subsection{Comparing complexity of list coloring and choosability}\label{subsec:remarks1} 

Problems $[\Lambda, k]$-CH and $[\Lambda, k]$-LISTCOL  appear to be close to each other and one may wonder about their relative complexity: is one of them more difficult than the other? It turns out that it is not the case. More precisely there are classes of instances where one of the problems is polynomially solvable while the other one is NP-hard. The aim of this subsection is to derive such examples. 

Consider the class ${\cal H}$ of subgrids containing a chocolate as a subgraph. Such a subgrid is trivially non $[\{2,3\},3]$-choosable since a chocolate itself is already non $[2,3]$-choosable (see~\cite{rubin}). Consequently $[\{2,3\},3]$-CH is trivially polynomial on the class ${\cal H}$ since the answer is NO for any of these instances.  On the other hand, we have shown in~\cite{grids}
 that $[\{2,3\},3]$-LISTCOL is NP-complete on subgrids of maximum degree~3 and the same result trivially holds on the class ${\cal H}$. 

Similarly, in~\cite{grids} we have given a subgraph of $G(9,9)$ with 41 vertices that is not $[(2,3),4]$-choosable and we have shown that $[(2,3),4]$-LISTCOL is NP-complete in grids and in particular in grids $G(p,q), p,q\geq 9$.  Using the same  argument as above in the class ${\cal H}$,  $[(2,3),4]$-CH is trivial for the class of grids $G(p,q), p,q\geq 9$ since these graphs are never $[(2,3),4]$-choosable. 

These simple examples give evidence that, in some classes of graphs $[\Lambda, k]$-LISTCOL is harder than $[\Lambda, k]$-CH. The following proposition gives an example of the reverse situation:

\begin{proposition}\label{prop: co-np}
There is a class ${\cal J}$ of graphs with known chromatic number $\chi$ for which $[\{\chi-2,\chi-1, \chi\},\chi]$-LISTCOL is polynomial while $[\{\chi-2,\chi-1, \chi\},\chi]$-CH is co-NP-complete.
\end{proposition}

\begin{proof}

\begin{figure}[h]
\begin{center}
\includegraphics[scale=0.4]{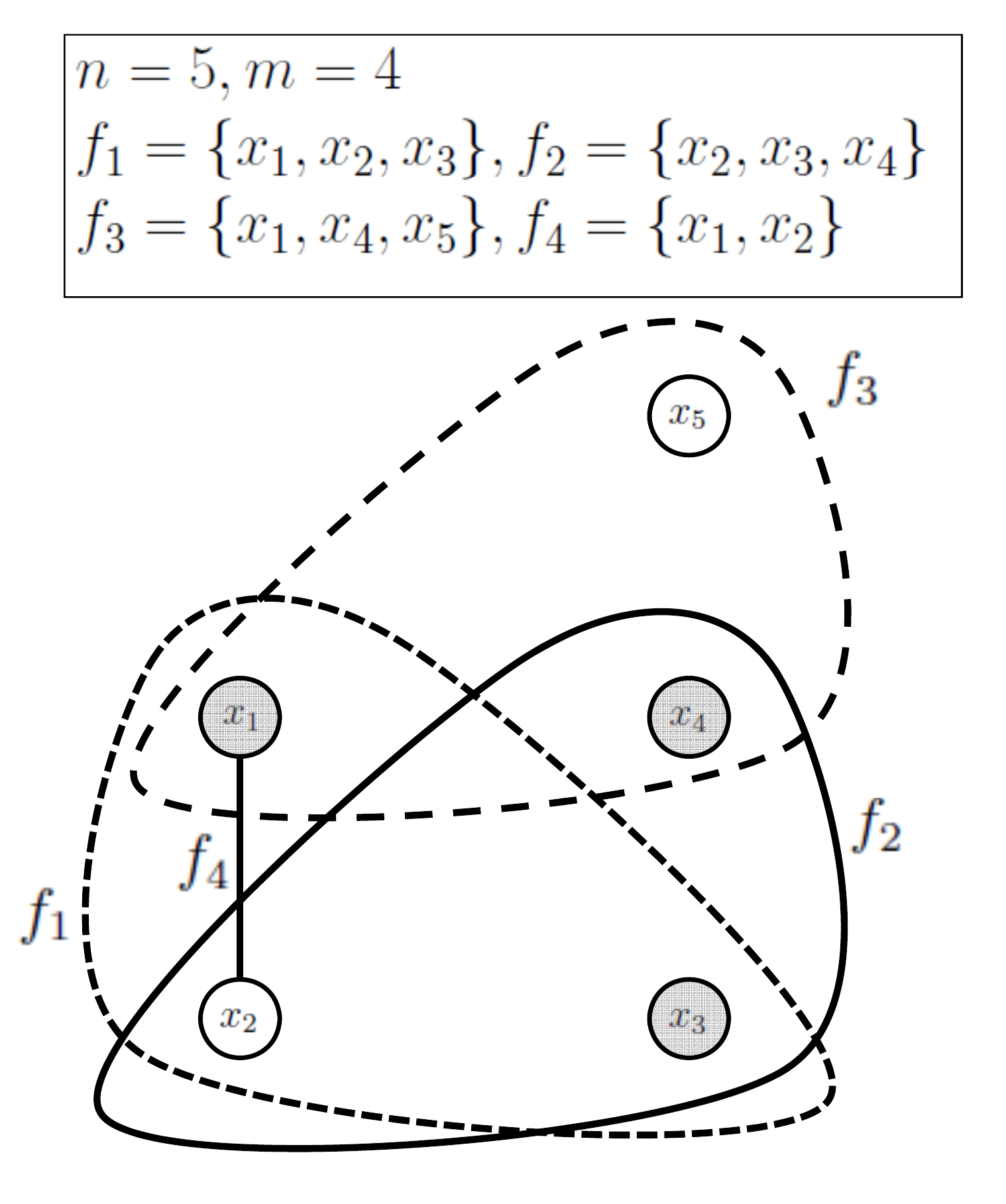}
\caption{A 2-colorable hypergraph $(X,F)$.}
\label{fig:hyper}
\end{center}
\end{figure}

\begin{figure}[h]
\begin{center}
\includegraphics[scale=0.5]{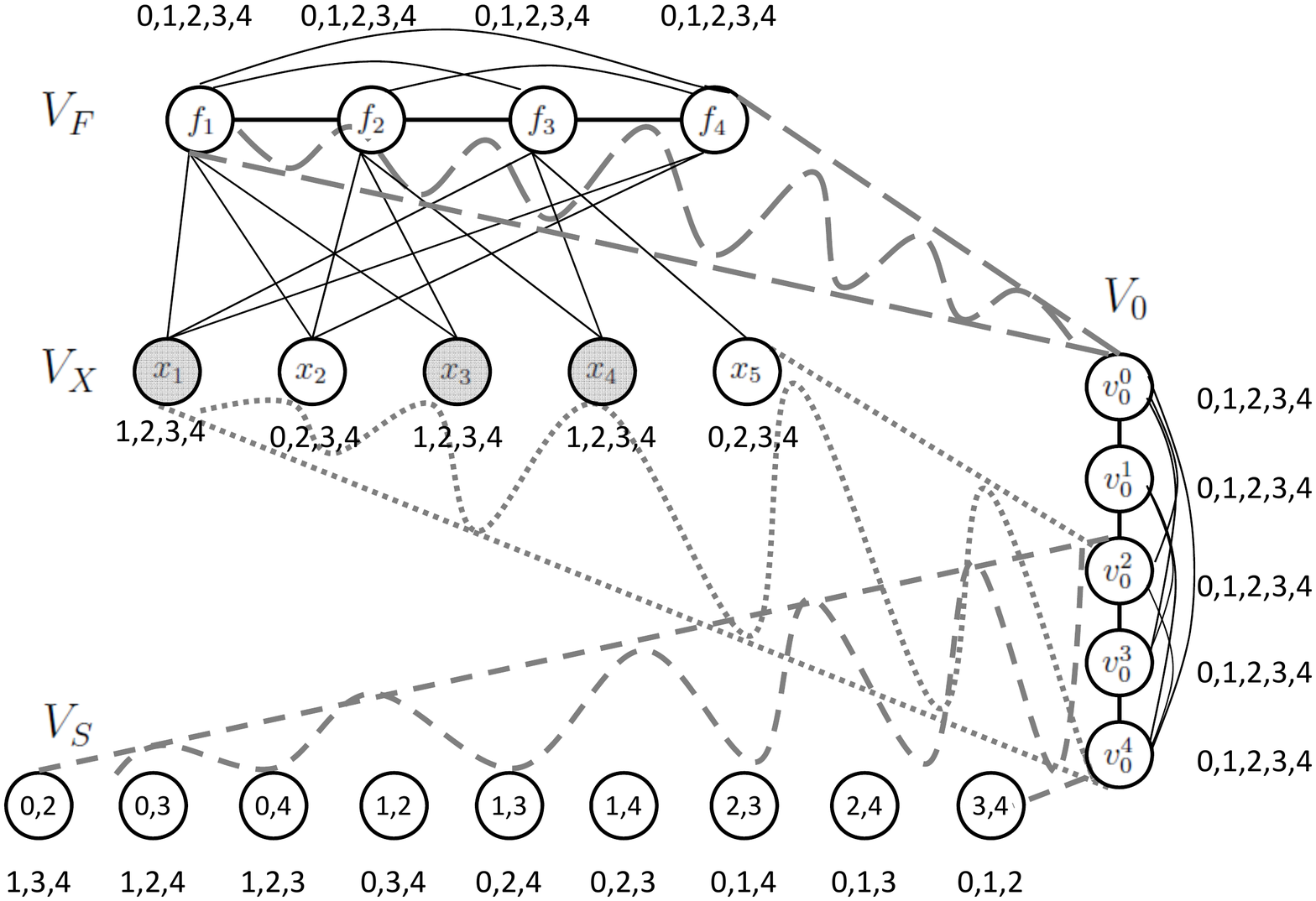}
\caption{The graph associated with $(X,F)$ of Figure~\ref{fig:hyper} and the related infeasible list assignment. }
\label{fig:CH-LCOL}
\end{center}
\end{figure}

We will use  Hypergraph 2-colorability (also called Set Splitting): any instance is defined by a hypergraph $(X,F)$ with vertex set $X=\{x_1, \ldots, x_n\}$ and hyperedge set $F=\{f_1, \ldots, f_m\}, f_i\subset X, i=1, \ldots m$; the question is whether there is a 2-coloring of $X$  such that no hyperedge is monochromatic. It is known to be NP-complete even if hyperedges are of size at most~3~\cite{gj}.

Let $(X,F)$ be such an instance with $n$ vertices and $m$ hyperedges. 
We consider the graph $G=(V,E)$ defined as follows: $V=V_0\cup V_F\cup V_X\cup V_S$, $V_0=\{v_0^0, v_0^1, \ldots, v_0^{m}\}$ induces a clique $K_{m+1}$, $V_F=\{v_2^{f_1}, \ldots, v_2^{f_m}\} $ induces a clique $K_m$, $V_X=\{v_X^{x_1}, \ldots v_X^{x_n}\}$  and $V_S=\{v_S^1, \ldots, v_S^{\binom{m+1}{2}-1}\}$. $V_X\cup V_S$ induces a stable set. We have in $E$ additional edges:  all edges between $v_0^2, \ldots, v_0^{m}$ and $V_X\cup V_S$, all edges between $v_0^0$ and $V_F$ and edges $v_F^{f_i}v_X^{x_j}$ for every vertex $x_j$ in $f_i$, $i\in\{1, \ldots m\}$. Finally vertices in $V_X$ will have $m$-lists, those in $V_S$ have $(m-1)$-lists and all other vertices $(m+1)$-lists. Let ${\cal J}$ be the class of such graphs.

$G$ can be constructed in polynomial time and it is straightforward to verify that its chromatic number is $\chi(G)=m+1$. There is indeed a clique of size $m+1$ and moreover an $(m+1)$-coloring can be obtained  as follows:  color vertex $v_0^i$ with color $i\in \{0, \ldots, m\}$, vertex $v_F^{f_i}$ with color $i\in\{1, \ldots m\}$ and vertices in $V_X\cup V_S$ with  0. 

Since $|V_S|=\binom{m+1}{2}-1$, there is a one-to-one correspondence between $V_S$ and the set of pairs of colors taken among $\{0, \ldots, m\}$, except the pair $\{0,1\}$; we denote by $(s_i,t_i)$ the pair corresponding to $i\in \{1, \ldots, \binom{m+1}{2}-1\}$.

\vspace{0,5 cm}

Figure~\ref{fig:hyper} gives an instance $(X,F)$ of Hypergraph 2-colorability with $n=5, m=4$. It is a 2-colorable hypergraph as shown by the black \& white coloring of $X=\{x_1, x_2, x_3, x_4, x_5\}$. Figure~\ref{fig:CH-LCOL} illustrates the construction of $G$ for the example of Figure~\ref{fig:hyper}. This graph is not $[\{3,4,5\},5]$-choosable as illustrated by the related infeasible list assignment which will be defined below. In this figure,  $V_X\cup V_F$ is directly identified with $X\cup F$ and vertices in $V_S$ are identified by  pairs of colors different from $(0,1)$, the black color is associated with 1 and the white one with 0. Finally, curved dashed lines indicate a complete link between two sets of vertices.

\vspace{0,5 cm}

Suppose that $(X,F)$ is 2-colorable and consider a feasible 2-coloring with colors 0 and 1. We then define a list system as follows: for every $i\in\{1, \ldots, n\}$, we take  $L(v_X^{x_i})= \{\theta_i\}\cup\{2, \ldots m\}$, where $\theta_i\in\{0,1\}$ is the color of $x_i$ in the 2-coloring of $(X,F)$.  For vertices in $V_S$ we define $L(v_S^i)=\{0, \ldots, m\}\setminus\{s_i,t_i\})$. In particular $L(v_S^i)$ contains at least one color among 0 and 1.  All other vertices have the list $\{0, \ldots, m\}$. In any feasible list coloring of $G[V_0\cup V_X\cup V_S]$ only two colors can be used for coloring $V_X\cup V_S$  and the list system of $V_S$ imposes that these two colors are necessarily 0 and 1. Consequently, any vertex in $V_X$ gets the same color as in the 2-coloring of $(X,F)$. Then at least one vertex  $v\in V_F$ should be colored either with 0 or 1 and since there is no monochromatic hyperedge, $v$ has two neighbors colored 0 and 1, a contradiction.  Consequently this list system is not feasible, proving that $G$ is not choosable.

\vspace{0,5 cm}

Suppose conversely that $G$ is not choosable for the above list sizes (i.e. $(m-1)$ in $V_S$, $m$ in $V_X$ and $(m+1)$ elsewhere) and consider an infeasible list system. Every vertex in $V_S$ has exactly two forbidden colors among $\{0, \ldots, m\}$  and since the total number of pairs of colors is  $\binom{m+1}{2}=|V_S|+1$ it is possible to find two colors for list coloring vertices in $V_S$. The same colors can be used for vertices in $V_X$ since each one has only one forbidden color.
Let us denote by $i,j$ the two colors used for coloring $V_X\cup V_S$. Color $v_0^2, \ldots, v_0^m$ with colors in $\{0, \ldots, m\}\setminus \{i,j\}$. Consider $v\in V_F$; if all its neighbors in $V_X$ are colored with the same color, say $i$, then it would be possible to extend the list coloring to the whole graph, choosing color $i$ for $v_0^0$,  color $j$ for $v_0^1$ and colors in $\{0, \ldots, m+1\}\setminus\{i,j\}$ for the vertices in $V_F\setminus \{v\}$, a contradiction with the fact that the list assignment is not feasible. 
Hence every $v\in V_F$ has a neighbor colored with $i$ in $V_X$ and one colored with $j$, this means that the coloring of $V_X$ defines a feasible 2-coloring of $(X,F)$. Hence, $(X,F)$ is 2-colorable if and only if $G$ is not choosable. 

Note finally that, given any list assignment with the corresponding list sizes, it can be checked in polynomial time whether $G$ is list colorable. Indeed, for every pair of colors $\{i,j\}$ among $\{0, \ldots, m\}$, one can color $v_0^0$ with $i$ and $v_0^1$ with $j$, then for every vertex in $V_F$, it suffices to check whether its neighbors can be colored with $i$ and, in this case, whether this coloring can be extended to $V_X\cup V_S$, using only colors $i$ and $j$. If so, the graph is list colorable. If this is not possible for any pair $\{i,j\}$ and any vertex in $V_S$, $G$ is not list colorable. This shows that 
$[\{\chi-2,\chi-1, \chi\},\chi]$-LISTCOL is polynomial in ${\cal J}$. 

Moreover this also shows that $[\{\chi-2,\chi-1, \chi\},\chi]$-CH is co-NP in the class ${\cal J}$ and consequently, the previous reduction shows that it is co-NP-complete.
\end{proof}

\subsection{Incidence on the complexity when adding one color in the palette}\label{subsec:remarks2} 

For list coloring problem, our definition of $[\Lambda,k]$-LISTCOL ensures that, $[\Lambda,k]$-LISTCOL is a subproblem of   $[\Lambda,k']$-LISTCOL for any $k'$ such that $\max(\Lambda)\leq k\leq k'$ and consequently if the former is hard on a class of instances, then the latter is also hard in this class. So, extending the palette cannot make the list coloring problem easier. 

Using  arguments similar to those used in the proof of Theorem~\ref{prop: 3-choos 4 col} we can easily exhibit an example where it makes it harder. Indeed, we have shown in~\cite{grids}  that  $[(2,3),4]$-LISTCOL is NP-complete in grids. Using the gadget $H$ in Figure~\ref{fig: 3-choos 4 col} we can show the following:

\begin{proposition}\label{prop:liscol}
$[3,4]$-LISTCOL is NP-complete in the class of 3-colorable planar graphs with a given 3-coloring.
\end{proposition}

\begin{proof}
The problem is clearly in NP as a special list coloring problem.
Take a palette $\{1,2,3,4\}$  and consider a grid graph $G$ with a $(2,3)$-list assignment,  instance of $[(2,3),4]$-LISTCOL. For any vertex $v$  with a 2-list, consider without loss of generality that  $3\notin L(v)$ and add a gadget $H_v$ isomorphic to $H$ with vertices $X,Y,Z$ linked to  $v$. Denote by $G'$ the new graph. Add color 3 to $L(v)$, set $L(X)=\{2,3,4\}$, $L(Y)=\{1,3,4\}$, $L(Z)=\{1,2,3\}$ and set the lists of the four other vertices of $H_v$ to $\{2,3,4\}$. Since the list assignment given in Figure~\ref{fig: 3-choos 4 col} is not feasible, in any list coloring of the new graph, at least one among $X,Y,Z$ needs to be colored with color~3 and consequently $v$ cannot be colored with~3. Moreover, by coloring $X,Y,Z$ with color~3, we can easily complete the list coloring on the whole gadget $H_v$. As a consequence, with these lists, a list coloring of $G'$ immediately defines a list coloring of $G$ and conversely a list coloring of $G$ can immediately be  extended to a list coloring of $G'$. If $3\in L(v)$, we can do a similar construction  changing the lists by a circular permutation. Repeating the construction for every vertex $v$ with a 2-list we build in polynomial time a graph $\widetilde G$ as well as a 3-list assignment such that $\widetilde G$ is list colorable if and only if $G$ is list colorable. Moreover, since $G$ is bipartite, $\widetilde G$ is 3-colorable and a 3-coloring can be  defined as follows: color $G$ with colors 1 and 2, color all vertices $X,Y,Z$ of gadgets $H-v$ with color~3 and complete the 3-coloring with colors 1 and 2 for the other vertices of gadgets $H_v$. This completes the proof using the fact that $[(2,3),4]$-LISTCOL is NP-complete in grids~\cite{grids}.
\end{proof}

Of course $[3,3]$-LISTCOL \-- equivalent to 3-coloring \-- is trivial in the class of 3-colorable planar graphs with a given 3-coloring and following Proposition~\ref{prop:liscol} makes the problem NP-complete.\\ 

A natural question arises: when increasing the number of colors does it make the problem of deciding whether a graph is choosable easier or even harder? 
It turns out that the situation is completely different from the case of List coloring since both situations may occur. 

In some cases, adding one color makes the problem harder. A trivial example is given by Theorem~\ref{prop: 3-choos 4 col}: $[3,4]$-choosability has been shown hard in a class of 3-colorable graphs, while $[3,3]$-choosability is trivial on this class. A similar situation is exhibited in Proposition~\ref{prop: l-choos k col bip} between $[3,4]$-choosability and $[3,5]$-choosability of bipartite graphs.


On the contrary, in some other cases, adding one color makes the problem easier.
In~\cite{old version}, we have  built a bipartite graph $H_5$ that is $[3,5]$-choosable but not $[3,6]$-choosable (See also Remark~\ref{rem:kk+1} and \cite{kral}). Consider the class obtained by adding to any bipartite graph an independent bipartite graph $H_5$. This transformation does not change the $[3,5]$-choosability while making the new graph not $[3,6]$-choosable. So, using Proposition~\ref{prop: l-choos k col bip}, $[3,5]$-CH is hard on this class while $[3,6]$-CH is trivial since the answer is always negative.

\section{Final remarks}\label{sec:conclusion} 

In this paper we started to investigate the complexity of some choosability problems when the size of the palette of colors is fixed. Table~\ref{table} summarizes our main hardness results.


\begin{center}
\begin{table}[h]
\centering
\footnotesize
\begin{tabular}{|c >{\centering\arraybackslash}p{1.8 cm}|c|>{\centering\arraybackslash}m{2.4 cm}|>{\centering\arraybackslash}m{3.4 cm}|}

\hline 
Problem & Graph class & $k=3$  & $k=4$ & $k=5$\\

\hline
$[\{2,3\},k]$-CH & Subgrids & \multicolumn{3}{|c|}{$\Pi_2^p$-complete (Thm.~\ref{prop: 2,3-choos 4 col grids})}\\
\hline
$[\{2,3,5\},k]$-CH & Grids &\multicolumn{2}{|c|}{\diagbox[dir=SW,width=4cm]{}{}}&$\Pi_2^p$-complete (Cor.~\ref{cor:grid 2,3,5})\\
\hline
\multirow{4}{*}{$[3,k]$-CH} & {3-colorable planar graphs} & \multirow{2}{*}{trivial} & \multicolumn{2}{c|}{\multirow{2}{*}{$\Pi_2^p$-complete (Thm.~\ref{prop: 3-choos 4 col})}} \\
\cline{2-5}
&{Triangle-free planar graphs} & \multirow{2}{*}{trivial} & \multirow{2}{*}{?} &\multirow{2}{*}{$\Pi_2^p$-complete  (Thm.~\ref{prop: 3-choos 5 col planar})} \\
\hline

\end{tabular}
\caption{Complexity status of choosability problem with 3,4 or 5 palette size}
\label{table}
\end{table}
\end{center}

This table shows some cases that still remain open to our knowledge. The main one is the complexity status of  $[3, 4]$-CH in triangle-free planar graphs; we conjecture that it is $\Pi_2^p$-complete but were still not able to prove it.

Another interesting question is  the complexity status of $[(2,3), 4]$-CH in bipartite graphs and more generally of $[(2,3), k]$-CH, $k\geq 4$: every bipartite graph is $[(2,3), 3]$-Choosable and $[(2,3), 4]$-LISTCOL is hard, even in subgrids~\cite{grids}. Also, the gadget we use to study the case of subgrids has vertices of degree~4; a possible direction to strengthen our results in subgrids would be to consider the case of subgrids of maximum degree~3.

Let us finally mention, as another research direction the edge choosability of graphs and list edge-coloring in specific classes of graphs when the number of colors is limited.

\newpage
\section*{Appendix}

In this appendix, we prove Lemma~\ref{lem_annex} which is needed for Theorem~\ref{prop: 2,3-choos 4 col grids}.

We first state some properties of the gadgets we have specified in the proof of Theorem~\ref{prop: 2,3-choos 4 col grids}. These properties are inspired by~\cite{rubin} and~\cite{gutner}.

$P_{1/2}$ has the following properties:
\begin{quote}
 {\bf (1)} for every $f_{P_{1/2}}$-list assignment and any choice of color $c\in L(I)$, there is a list coloring with $I$ colored~$c$,\\
 {\bf(2)} for any choice of color $c'\in L(O')$, there is a list coloring of $ P_{1/2}$ with $O'$ colored~$c'$,\\
 {\bf(3)} there is an $f_{P_{1/2}}$-list assignment with $0\in L(I)$ and $0\in L(O')$ and such that in any possible list coloring with $I$ colored $0$, $O'$ will also be colored $0$. 
\end{quote}

 (1) is immediate by using Proposition~\ref{claim: C} and ordering vertices as in Figure~\ref{fig:numbers}.
 
 \begin{figure}[h]
\begin{center}
\includegraphics[scale=0.3]{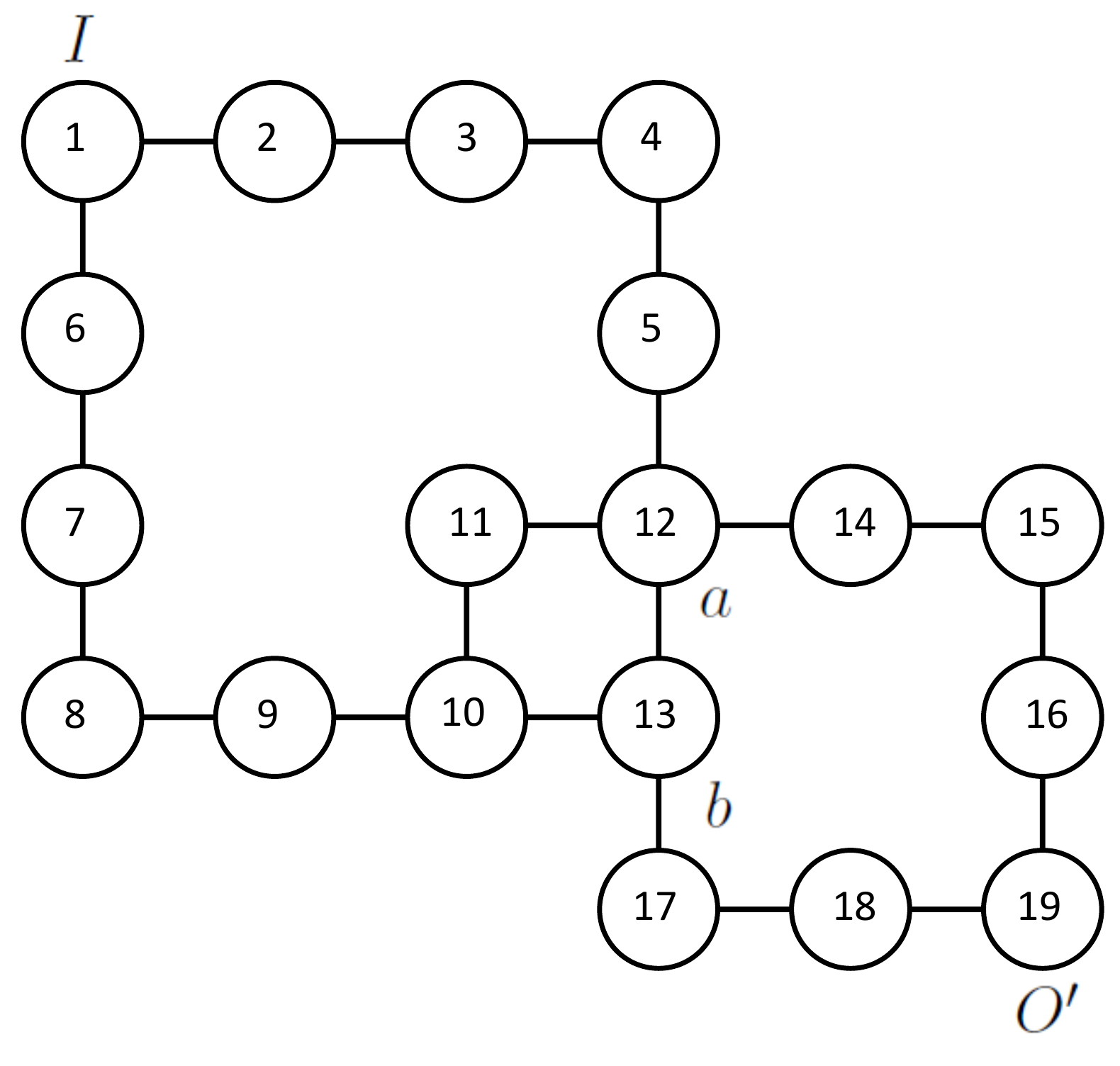}
\caption{A numbering of $P_{1/2}$.}
\label{fig:numbers}
\end{center}
\end{figure}

(2) is a consequence of the fact that $\theta_{2,2,10}$ is 2-choosable.
(3) can be shown using the $f_{P_{1/2}}$-list assignment given in Figure~\ref{fig:half-propagator}.

An immediate consequence  of (2) is:
\begin{quote}
{\bf(4)} for any $f_{P_{1/2}}$-list assignment, there is a color $c\in L(I)$ such that for two different colors $c',d'\in L(O')$ there are two list colorings with $I$ colored $c$ and $O'$ colored $c'$ and $d'$, respectively. 
\end{quote}

For a gadget $P$, if one fixes the color $c$ of its input vertex $I$, a color $d$ will be  {\em available} for the output vertex $O$ if there is a list coloring of $P$ such that $I$ is colored $c$ and $O$ is colored $d$. Note that the colors of the other vertices of $P$ will not affect the rest of the graph.
 
Then, the gadget $P$ with input vertex $I$ and output vertex $O$ will satisfy the following properties:
\begin{quote}
{\bf(1')} for every $f_{P}$-list assignment and any choice of color $c\in L(I)$, there is a list coloring with $I$ colored~$c$, 
 
{\bf (3')} there is an $f_{P}$-list assignment with $0\in L(I)$ and $0\in L(O)$ and such that in any possible list coloring with $I$ colored $0$, $O$ will also be necessarily colored~$0$,\\
{\bf (4')} for any $f_{P}$-list assignment, there is a color $c\in L(I)$ such that all three colors in $L(O)$ are available. We then call $c$ a {\em nice} color in $L(I)$.	
\end{quote}
To show (4'), one selects $c\in L(I)$ such that it is compatible with two different colors in $L(O')$ and then we apply property (2) on the second $P_{1/2}$ forming $P$ with $L(O')$ reduced to these two colors.  
	
	We will  finally  need two additional easy properties:
\begin{quote}
{\bf(5)} For the gadget associated with the $\forall$-variable $U_i$ shown in Figure~\ref{fig:forall}, for any 2-list assignment, one of the two literal vertices $u_i, \bar u_i$  may have only one possible color in any list coloring, but in this case the other one can be given any of its two possible colors,\\
{\bf(6)} For an elementary path from $x$ to $y$ of length $p\geq 2$, there is a 2-list assignment with palette $\{0,1,2\}$ and $L(x)=\{0,1\}$ such that for any $i\in\{0,1,2\}$ if $x$ is colored 0 then $y$ gets color $i$.
\end{quote}

Property (5) has already been noted in~\cite{rubin}.  Note first that the gadget associated with the $\forall$-variable is 2-choosable (its core is a $C_4$). An example of a 2-list assignment imposing  color 1 to $u_i$ is proposed in Figure~\ref{fig:forall}. Consider now any 2-list assignment imposing one color at a literal vertex $w\in\{u_i,\bar u_i\}$. Without loss of generality assume $w=u_i$, $L(u_i)=\{0,1\}$ and 1 is imposed, which means there is no list coloring with $u_i$ colored 0. Then necessarily $L(u_i^1)=\{0,c\}, c\neq 0$ and $c\in L(u_i^2)\cap L(u_i^4)$. Indeed, if $0\notin L(u_i^1)$ any 2-list coloring can be changed to force color 0 for $u_i$ and if $L(u_i^1)=\{0,c\}$ but $c\notin L(u_i^2)\cap L(u_i^4)$, say $c\notin L(u_i^2)$, we can assign color 0 to $u_i$, color $c$ to $u_i^1$ and then complete greedily the list coloring considering vertices in order $u_i^4, u_i^3, u_i^2, \bar u_i$. By symmetry the same can be shown if $c\notin L(u_i^4)$ or if $w=\bar u_i$. \\
For property (6) we consider the three following cases:\\
For $i=0$ and $p$ even all lists are $\{0,1\}$.\\
For $i=0$ and $p$ odd lists are 
$\{0,1\},\ldots \{0,1\},\{1,2\}, \{0,2\}$.\\
For $i\neq 0$, let $j$ be such that $\{0,i,j\}=\{0,1,2\}$ and lists are: 
$\{0,1\},\{0,j\},\ldots$ $\{0,j\},$ $\{\ell,i\}$ where 
$\ell= 0$ if $p$ is odd and $\ell=j$ if $p$ is even.
\vspace{0.5 cm}

{\underline{Proof of the reduction}}:

We are now ready to remind the main arguments for proving statements (i) and (ii).

{\bf (i)} Suppose that ${\cal I}$ is satisfiable and consider any $f$-list assignment for $\widetilde{\Gamma}_{\Phi}$.  Considering property (5), this list assignment may impose some colors among  the literal vertices $u,\bar u$ associated with $\forall$-variables, but in such a case the color of the vertex associated with the negation literal can be any color of its list. We set the related literals (with their color imposed) to False and call them {\em forced}. This defines a truth assignment for some $\forall$-variables. We  choose arbitrarily the truth assignment of the other $\forall$-variables. We then consider a truth assignment of $\exists$-variables such that each clause contains at least one true literal. 

Then for every   vertex $v$ associated with a true literal, we choose $c$, a nice color in its list  using (4') ($v$ is the input vertex of a gadget $P$), letting available all three colors of the output vertex. Note that it is always possible and in particular for the negation of forced literals using Property (5). If a true literal appears twice in $\Phi$ then, we choose a nice color of the input vertex of its second gadget $P$. In this case the output of the first gadget $P$ has at least two available colors and the output vertex of the second gadget $P$ has its three colors available.  We then complete the list coloring of all variable gadgets in particular  false literals of $\forall$-variables and false literals of $\exists$-variables that are not forced.  For any vertex associated with a false literal, the related gadget $P$ and, possibly the second one are list colored using property (1'). We also color the connection paths between the Outputs of these colored $P$'s and clause vertices. Then we color the clause vertices, which is possible since no clause vertex has all its three neighbors already colored (only those corresponding to a false literal are colored). Finally, we  color the connection paths associated with true literals from the clause vertex to the related output vertex of a gadget $P$. It is possible since this output vertex has at least two available colors. Using Property (4') it can be extended to the related $P$'s, which concludes the proof: $\widetilde{\Gamma}_{\Phi}$ is choosable.\\

{\bf (ii)} Suppose now $\widetilde{\Gamma}_{\Phi}$ is $[f,3]$-choosable; consider any truth assignment for $\forall$-variables and we will show there is a truth assignment for $\exists$-variables satisfying all clauses. We construct an $f$-list assignment for $\widetilde{\Gamma}_{\Phi}$ as follows. We first fix the 2-lists of $\forall$-variable gadgets so that the color of the False $\forall$-literal vertices is necessarily 0 and  the related True $\forall$-literal vertices get the list $\{0,1\}$, both compatible (see Figure~\ref{fig:forall}). The $\exists$-literal vertices have the list $\{0,1\}$ such that $0$ for one literal implies 1 for its opposite. We also determine the lists of all gadgets $P$ and intermediate vertices $M$ such that color $0$ for the Input induces  color 0 for all related Outputs (see Figure~\ref{fig:half-propagator}). Clauses are associated with the 3-list $\{0,1,2\}$. The literals of the clause are ordered $x_0,x_1,x_2$, which defines a numbering of the three connection paths arriving at this vertex (path $i$, $i=0,1,2$ corresponds to literal $x_i$). Then define 2-lists on the connection paths such that  color 0 at the Output vertex induces  color $i$ at the last vertex before the related clause vertex.

Since $\widetilde{\Gamma}_{\Phi}$ is $[f,3]$-choosable, there is a list coloring. It defines a truth assignment for all $\exists$-variables, color $0$ corresponding to the value False. By choice of the lists, a clause associated with only false literals could not be colored. This means that this truth assignment satisfies all clauses, which concludes the proof.

\end{document}